\documentclass[12pt]{article}    

\usepackage[margin=0.75in]{geometry} 
\usepackage{amsmath,amssymb,amsthm}
\usepackage{extarrows}

\newtheorem{theorem}{Theorem}[section]
\newtheorem{proposition}[theorem]{Proposition}
\newtheorem{lemma}[theorem]{Lemma}
\newtheorem{corollary}[theorem]{Corollary}

\theoremstyle{remark}
\newtheorem*{remark}{{\bf Remark}}

\newcommand{\al}{\alpha}

\newcommand{\be}{\begin{equation}}
\newcommand{\ee}{\end{equation}}

\newcommand{\tr}{{\rm tr}}

\newcommand{\non}{{\nonumber}}

\newcommand{\vol}{\rm vol}
\newcommand{\bea}{\begin{eqnarray}}
\newcommand{\eea}{\end{eqnarray}}

\numberwithin{equation}{section}
\begin{document}

\title{Asymptotic Gap Probability Distributions of the Gaussian Unitary Ensembles and Jacobi Unitary Ensembles}
\author{Shulin Lyu$^{a,}$\thanks{lvshulin1989@163.com} ,
Yang Chen$^{b,}$\thanks{yangbrookchen@yahoo.co.uk}\, and Engui Fan$^{c,}$\thanks{ Corresponding author and faneg@fudan.edu.cn}\\
\footnotesize{$^{a}$ School of Mathematics (Zhuhai Campus), Sun Yat-Sen University, Guangzhou, China}\\
\footnotesize{$^{b}$ Department of Mathematics, University of Macau, Avenida da Universidade, Taipa, Macau, China}\\
\footnotesize{$^{c}$ School of Mathematical Sciences and Key Laboratory of Mathematics for Nonlinear Science,}\\
\footnotesize{Fudan University, Shanghai 200433, China}}


\date{\today}
\maketitle
\begin{abstract}
In this paper, we address a class of problems in unitary ensembles.
Specifically, we study the probability that a gap symmetric about 0, i.e. $(-a,a)$  is found in the Gaussian unitary ensembles (GUE)
and the Jacobi unitary ensembles (JUE) (where in the JUE, we take the parameters $\alpha=\beta$).
By exploiting the even parity of the weight, a doubling of the interval to $(a^2,\infty)$ for the GUE, and $(a^2,1)$, for the (symmetric) JUE, shows
that the gap probabilities maybe determined as the product of the smallest eigenvalue distributions of the LUE with parameter $\alpha=-1/2,$
and $\alpha=1/2$ and the (shifted) JUE with weights $x^{1/2}(1-x)^{\beta}$ and $x^{-1/2}(1-x)^{\beta}$
\\
The $\sigma$ function, namely, the derivative of the log of the smallest eigenvalue distributions of the finite-$n$ LUE or the JUE,
satisfies the Jimbo-Miwa-Okamoto $\sigma$
form of $P_{V}$ and $P_{VI}$,  although in the shift Jacobi case, with the weight $x^{\alpha}(1-x)^{\beta},$ the $\beta$ parameter does not show up in the equation.
We also obtain the asymptotic expansions for the smallest eigenvalue distributions of the Laguerre unitary and Jacobi unitary ensembles
after appropriate double scalings, and obtained the constants in the asymptotic expansion of the gap probablities, expressed in term of the Barnes $G-$ function valuated at special point.
\end{abstract}

\section{Introduction}
In the work of Adler and Van Moerbeke \cite{Adler}, the largest eigenvalue distribution of ensembles of $n\times n$ random matrices generated by Gaussian, Laguerre and Jacobi weights for general values of the symmetry parameter $\beta,$ (not to be confused with the $\beta$ parameter in the Jacobi weight), has been systematically studied, from the perspective of differential operators involving multiple time variables.

The gap probabilities that are studied in this paper, the unitary case, denoted by $\mathbb{P}(a,n)$, are represented as  Fredholm determinant of an integral operator, from the early papers of Mehta, Gaudin and Dyson \cite{Dyson1962JMP}, \cite{DysonMehta1963}, \cite{Gaudin1961}, \cite{MehtaDyson1963JMP}, and \cite{MehtaGaudin1960}.

In \cite{JimboMiwaMoriSato1980}, the gap probability, where a union of disjoint intervals is free of eigenvalues, the integral operator has the sine kernel $K(x,y):=\frac{\sin(x-y)}{\pi (x-y)}$. The (multiple-) gap probability itself was obtained in an expansion in terms of the resolvent of the integral equation. In a tour de force computations, JMMS showed that in the single interval case where $(-b,b)$
is free of eigenvalues, the quantity $\sigma(\tau):=\tau\frac{d}{d\tau}{\rm log}{\rm det}(I-K_{(-b,b)}),~\tau=2b$, satisfies a second order non-linear differential equation.

Tracy and Widom \cite{TracyWidom1994} studied, the finite $n$ version of such problems, namely, the distribution of the smallest
eigenvalue in the Laguerre unitary ensembles, and the largest eigenvalue distribution of the Gaussian unitary ensembles starting from the
Christoffel-Darboux or Reproducing  Kernel, constructed out of the ``natural" orthogonal polynomials, namely the Laguerre and Hermite polynomials, respectively. Through a series of differentiation formulas, Tracy and Widom found the finite $n$ version of the differential equations satisfied by the resolvent and the related $\sigma$ form in the Laguerre and Hermite cases.

An elementary method to deal with $\mathbb{P}(a,n)$ is to write it as a Hankel determinant, or determinant of moment matrices where the moments are generated by the classical weight function, such as
${\rm e}^{-x^2}$, $x^{\alpha}{\rm e}^{-x},$ ${\it multiplied}$  by one minus the characteristic function of an interval $J$.
It is clear from such determinant representations that one could also study the ${\it non-standard}$ polynomials orthogonal with respect to
the classical weight multiplied by $(1-\chi_{(-a,a)}(x))$ instead of writing such quantities as a Fredholm determinant $\det(I-K_{(-a,a)})$. Here $\chi_{(-a,a)}(x)$ is characteristic function of the interval $(a,a)$.

In our approach, the Hankel determinants is expressed as the product of the square of the $L^2$ norms, $h_k(a),$ of the ${\it non-standard}$ orthogonal polynomials namely, $\prod_{k=0}^{n-1}h_k(a).$

Based on the ladder operators adapted to the ${\it non-standard}$ orthogonal polynomials, and from the associated supplementary conditions and a sum-rule, satisfied by certain rational functions (depending on the degree),  a series of difference and differential equations can be derived to give a description of $\mathbb{P}(a,n)$.
For detailed descriptions and applications of such formalism, see for example, \cite{BasorChenEhrhardt2010}, \cite{BasorChenZhang2012}, \cite{ChenChenJMP2015}, \cite{ChenChenFanJMP2016}, \cite{ChenFeigin2006}, \cite{ChenIts2010}, \cite{ChenMcKay2012}, and \cite{ChenPruessner2005}.

In \cite{CaoChenGriffin2014}, such an approach was taken to study the gap probability problem for the Gaussian unitary ensembles (the symmetric situation), namely, the probability that the interval $J:=(-a,a)$ is free of eigenvalues.
Unfortunately the authors have made a mistake: One term was missed in an equation obtained from the sum-rule. We  present the correct version here, but not the derivation and refer the readers to algebraically
more complicated case where the back-ground weight is the symmetric Beta density $(1-x^2)^{\beta}$,$\;\;x\in[-1,1]$ \cite{MinChen2017}. That is,  we study polynomials orthogonal
with respect to the deformed weight $(1-x^2)^{\beta}(1-\chi_{(-a,a)}(x)),\;\;a<1,$ where for convenience we take $\beta>0.$
Here and what follows, $\chi_{(-a,a)}(x)$ is the characteristic (or the indicator)function of the interval $(-a,a)$, namely $\chi_{(-a,a)}(x)=1$ if $x\in(-a,a)$
and $\chi_{(-a,a)}(x)=0,\;$ if $x\notin(-a,a)$.
We note that our approach is entirely elementary, up to some distributional objects. If $\theta(x)$ is the step function, which takes value 1 if $x>0$ and 0 if $x\leq 0,$ then $\frac{d}{dx}\theta(x)=\delta(x),$ the Dirac Delta.

In the case of the ${\it symmetric}$ Jacobi unitary ensembles generated by the weight $(1-x^2)^{\beta},~|x|<1$. The
probability that a gap $(-a,a),~|a|<1$, is formed, can also be found via an ``interval doubling", exploiting the parity of
the orthogonal polynomials, to the arrive at the Hankel determinants generated by the shifted Jacobi weight $x^{\alpha}(1-x)^{\beta},$ over the
interval $[t,1].$
Although the asymptotic expansion of the large gap probability can be relatively easily obtained, the determination of the constant of integration is  not that straightforward.

In our derivation, we will make use of Dyson's Coulomb fluid approximation \cite{Dyson1962JMP}.
We give a brief description of it. A unitary ensemble of Hermitian matrices $M=(M_{ij})_{n\times n}$ has probability density
\begin{equation*}
\begin{aligned}
&p(M)dM\propto {\rm e}^{-\tr\;v(M)}\vol(dM),\\
&{\vol(dM)}=\prod\limits_{i=1}^{n}dM_{ii}\prod\limits_{1\leq j<k\leq n}d({\rm Re} M_{jk}) d ({\rm Im} M_{jk}).
\end{aligned}
\end{equation*}
Here $v(M)$ is a matrix function \cite{Higham2008} defined via Jordan canonical form and $\vol(dM)$ is called the volume element \cite{Hua1963}.
Under an eigenvalue-eigenvector decomposition, the joint probability density function of the eigenvalues $\{x_k\}_{k=1}^n$ of this ensemble is given by \cite{Mehta2006}
\begin{gather}\label{jpdf}
\frac{1}{D_n}\;\frac{1}{n!}{\prod\limits_{1\leq i<j\leq n}\left(x_{j}-x_{i}\right)}^{2}\prod\limits^{n}_{k=1}w(x_k),
\end{gather}
where $w(x):={\rm e}^{-v(x)}$
is a positive weight function supported on $[A,B]$ with finite moments
$$\mu_k:=\int_A^B x^{k}w(x)dx,\qquad k=0,1,2,\cdots.$$
The normalization constant $D_n$ can be evaluated as the determinant of the Hankel (or moment) matrix (see \cite{Szego1939}), i.e.
\begin{align*}
D_n:=&\frac{1}{n!}\int_{[A,B]^{n}}{\prod\limits_{1\leq i<j\leq n}\left(x_{j}-x_{i}\right)}^{2}\prod\limits^{n}_{k=1}w(x_k)dx_{k},\\
=&\det\left(\mu_{i+j}\right)_{i,j=0}^{n-1}.
\end{align*}


If we interpret $\{x_k\}_{k=1}^n$ as the positions of $n$ identically charged particles,
then, for sufficiently large $n$, the particles can be approximated as a continuous fluid with a density $\rho(x)$.
We assume $v(x):=\log w(x)$ is convex, so that $\rho(x)$ is supported on a single interval $(a,b)$. Note that $(a,b)$ has nothing to do with $(A,B)$.
See \cite{ChenIsmail1997JPA:MG} for a detailed analysis.
Such $\rho(x)$ is determined by minimizing the functional
\begin{gather*}
F[\rho]:=\int_a^b \rho(x)v(x)dx-\int_a^b\int_a^b\rho(x)\log|x-y|\rho(y)dxdy,
\end{gather*}
subject to
\begin{gather*}
\int_a^b \rho(x)dx=n.
\end{gather*}
See Dyson's works \cite{Dyson1962JMP}.
According to Frostman's lemma \cite{Tsuji1959}, the equilibrium density $\rho(x)$ satisfies the condition
\begin{gather*}
v(x)-2\int_a^b \log|x-y|\rho(y)dy=\textbf{A},{\qquad} x\in[a,b],
\end{gather*}
where $\textbf{A}$ is the Lagrange multiplier that fixes the condition $\int_{a}^{b}\rho(x)dx=n.$
The derivative of this equation with respect to $x$ gives rise to the singular integral equation
\begin{gather*}
2\:P\int_a^b \frac{\rho(y)\,dy}{x-y}=v'(x),\qquad x\in[a,b],
\end{gather*}
where $P$ denotes the Cauchy principal value.
According to the theory of singular integral equations \cite{Mikhlin1957}, we find
\begin{gather}\label{densitycon.}
\rho(x)=\frac{1}{2\pi^2}\sqrt{\frac{b-x}{x-a}}\,\mathcal{P}\int_a^b \frac{ v'(y)}{y-x}\sqrt{\frac{y-a}{b-y}}dy.
\end{gather}
Hence, the normalization condition $\int_a^b \rho(x) dx=n$ becomes
\begin{gather*}
\frac{1}{2\pi}\int_a^b \sqrt{\frac{y-a}{b-y}}v'(y)dy= n.
\end{gather*}

Based on the above Coulomb fluid interpretation and the notion of linear statistics, it is proved in \cite{ChenLawrence1998}, for sufficiently large $n$, that the monic polynomials orthogonal with respect to $w(x)={\rm e}^{-v(x)}$ can be approximated as follows:
\begin{subequations}\label{PnS1S2general}
\begin{equation}
P_n(z)\sim {\rm e}^{-S_1(z)-S_2(z)},\qquad z\in\mathbb{C}\backslash[a,b],
\end{equation}
where
\begin{equation}
S_1(z)=\frac{1}{4}\log\left[\frac{16(z-a)(z-b)}{(b-a)^2}\left(\frac{\sqrt{z-a}-\sqrt{z-b}}{\sqrt{z-a}+\sqrt{z-b}}\right)^2\right],
\end{equation}
\begin{equation}
\begin{aligned}
S_2(z)=&-n\log\left(\frac{\sqrt{z-a}+\sqrt{z-b}}{2}\right)^2\\
&+\frac{1}{2\pi}\int_a^b \frac{v(x)}{\sqrt{(b-x)(x-a)}}\left[\frac{\sqrt{(z-a)(z-b)}}{x-z}+1\right]dx.
\end{aligned}
\end{equation}
This is a simpler representation for ${\rm e}^{-S_1}$:
\begin{equation}
{\rm e}^{-S_1(z)}=\frac{1}{2}\left[\left(\frac{z-b}{z-a}\right)^{\frac{1}{4}}+\left(\frac{z-a}{z-b}\right)^{\frac{1}{4}}\right].
\end{equation}
\end{subequations}

This paper is organized as follows. In section 2, we give a quick summery of what was know regarding the gap probability of GUE, and finish with
an elementary identity expressing the desired Hankel determinant as the products of Hankel determinants generated by $x^{\pm 1/2}\;{\rm {exp}}(-x),$
over $a^2<x<\infty.$
Section 3 is devoted to the computation of the smallest eigenvalue distribution of the LUE and we obtain an asymptotic expansion of the large gap probability, including the hard-to-come-by constant term. In section 4,
we note the elementary fact that for any polynomials orthogonal with respect to an even weight, a doubling process which folds interval, for example, from $(-\infty,\infty)$ to $(0,\infty)$, and $(-1,1)$ to $(0,1)$, transforms the problem with two discontinuities, (due to $\chi_{(-a,a)}(x)$)  to problems with one gap. This simplifies things considerably. From these ${\it one~gap}$ problems, combined with the large $n$ asymptotic expansion of the ${\it deformed}$ orthogonal polynomials, we compute the constants, that appear in the asymptotic expansion. We investigate the smallest eigenvalue distribution
of JUE in section 5, and found the asymptotic expansion for large gap probability, together with the constant term. We present in section 6, the asymptotic gap probability of the (symmetric) Jacobi ensembles where the back ground weight reads $(1-x^2)^{\beta},~|x|<1. $
\section{Gap Probability of the Gaussian Unitary Ensembles}
The weight function reads
\begin{gather*}
w_{0}(x)=\mathrm{e}^{-x^2},\qquad x\in\mathbb{R}.
\end{gather*}
Define
\begin{gather*}
w(x,a):=w_{0}(x)(1-\chi_{(-a,a)}(x)).
\end{gather*}
Hence $w(x,a)$ is ${\rm e}^{-x^2}$ for $x\in(-\infty,-a)\bigcup(\infty,a).$

According to \eqref{jpdf} and the theory of orthogonal polynomials \cite{Mehta2006},
we know that the probability that $(-a,a)$ is free of eigenvalues in the Gaussian unitary ensembles is given by
\begin{align*}
\mathbb{P}(a,n)=&\frac{\int_{\mathbb{R}^n}\frac{1}{n!}{\prod\limits_{1\leq i<j\leq n}\left(x_{j}-x_{i}\right)}^{2}\prod\limits^{n}_{k=1}w(x_k,a)dx_k
}{\int_{\mathbb{R}^n}\frac{1}{n!}{\prod\limits_{1\leq i<j\leq n}\left(x_{j}-x_{i}\right)}^{2}\prod\limits^{n}_{k=1}w_0(x_k)dx_k
}\\
=&
\frac{\det\left(\int_{\mathbb{R}}x^{i+j}w(x,a)dx\right)_{i,j=0}^{n-1}}
{\det\left(\int_{\mathbb{R}}x^{i+j}w_{0}(x)dx\right)_{i,j=0}^{n-1}}\\
=:&\frac{D_n(a)}{D_n(0)}=\frac{\prod_{j=0}^{n-1}h_j(a)}{\prod_{j=0}^{n-1}h_j(0)}.
\end{align*}
Here $h_j(a)$ is the square of the $L^2$ norm of the monic polynomials orthogonal with respect to $w(x,a)$:
\begin{gather}\label{orth.poly.}
\int_{\mathbb{R}}P_{j}(x;a)P_{k}(x;a)w(x,a)dx=h_{j}(a)\delta_{jk},\;\;j,k=0,1,2,\ldots.
\end{gather}
where $P_{j}(x;a)$ can be normalized (since the weight is even) as \cite{Chihara}
\begin{gather*}
P_{j}(x;a)=x^{j}+\mathrm{p}(j,a)x^{j-2}+\cdots+P_{j}(0;a).
\end{gather*}
In addition, note that $D_n(0)$ has the following explicit expression (see \cite{Mehta2006}, p. 321)
\begin{gather*}
D_n(0)=(2\pi)^{n/2}2^{-n^2/2}G(n+1),
\end{gather*}
with $G(\cdot)$ denoting the Barnes-G function, defined by the functional equation\\
$G(z+1)=\Gamma(z)G(z),$ with $\;G(1)=1.$ See \cite{Voros1987} for a detail description.
%
We shall see that $\mathrm{p}(n,a)$ plays an important role.

\begin{remark}
The dependence of $P_n(x;a)$ on $a$ is seen from its determinant representation in terms of the moments,or the Heine Formula, see \cite{Szego(1939)} eq.(2.2.10).
\begin{align*}
P_{n}(z;a)=&\frac{1}{D_n(a)}\frac{1}{n!}\int_{\mathbb{R}^{n}}
\prod\limits_{1\leq i<j\leq n}(x_{j}-x_{i})^2\prod\limits^{n}_{k=1}(z-x_k)w(x_k,a)dx_k\\
=&\frac{\det \left(\int_\mathbb{R}x^{i+j}(z-x)w(x,a)dx\right)_{i,j=1}^{n-1}}{\det \left(\int_\mathbb{R}x^{i+j}w(x,a)dx\right)_{i,j=1}^{n-1}}.
\end{align*}
\end{remark}
\begin{remark}
It should also be pointed out that $P_n(z;a) $ contains only the terms $x^{n-j}$, $j\leq n$ and even, since the weight function $w(x,a)$ defined on $\mathbb{R}$ is even. This implies that
\begin{gather*}
P_n(-x;a)=(-1)^n P_n(x;a),\qquad P_n(0;a)P_{n-1}(0;a)=0.
\end{gather*}
\end{remark}

\begin{remark} On could have written the gap probability as $\det(\mathbb{I}-\mathbb{K}_(a,a))$ as in per \cite{JimboMiwaMoriSato1980}, where $K$ is the Christoffel-Darboux kernel which acts on functions as follows:
$$\int_{-a}^{a}\mathbb{K}_n(x,y)f(y)dy.$$ In this setting the kernel is the re-producing kernel
constructed out of the ``free" or ``unperturbed" orthogonal polynomials. For the finite $n$ problem it is the Hermite polynomials. Under ``double scaling", see Theorem 2.2, this becomes becomes the sine kernel.
\end{remark}

\subsection{The $\sigma$ form of Painlev\'{e} V}
We state here
the difference and differential equations satisfied by
\begin{gather*}
\sigma_n(a):=a\frac{d}{da}\log \mathbb{P}(a,n)=a\frac{d}{da}\log D_n(a).
\end{gather*}
Note that $\mathbb{P}(a,n)=\tfrac{D_n(a)}{D_n(0)}$.

\begin{theorem} 
For a fixed $a$, $\sigma_n$ satisfies the second order difference equation
\begin{equation}\label{sigmandiff}
\begin{aligned}
&\left[(\sigma_{n-1}-\sigma_{n})\left(\sigma_{n}-\sigma_{n+1}-2 a^2\right)\left((\sigma_{n-1}-\sigma_{n}-2a^2) (\sigma_{n}-\sigma_{n+1})+4n a^2\right)\right.\\
&\left.\qquad-8 a^4 (\sigma_{n}+n (\sigma_{n}-\sigma_{n+1}))\right]^2\\
   &=(\sigma_{n-1}-\sigma_n) (\sigma_{n}-\sigma_{n+1}) \left(\sigma_{n-1}-\sigma_{n}-2 a^2\right)^2
   \left(\sigma_{n}-\sigma_{n+1}-2 a^2\right)^2\\
   &\qquad\cdot\left((\sigma_{n-1}-\sigma_{n}) (\sigma_{n}-\sigma_{n+1})+8n a^2\right).
\end{aligned}
\end{equation}
and the following second order fourth degree differential equation:
\begin{equation}\label{sigmanODE}
\begin{aligned}
&16\left[a^2\sigma_n ''+4(\sigma_n+2na^2)(a \sigma_n '-\sigma_n-a^4)-4a^4\right]\\
&\qquad\cdot\left[a^4(\sigma_n '')^2-4 a^2 \sigma_n''\left(a \sigma_n '-\sigma_n \right)+4 \left(a \sigma_n '-\sigma_n-a^4\right) \left(a
   \sigma_n '-2 \sigma_n \right)^2\right]\\
   &=\left[a^2(\sigma_n '')^2+4\left((\sigma_n')^2+8n(a\sigma_n '-\sigma_n)-4a^2\right)(a \sigma_n '-\sigma_n-a^4)-16a^6\right]^2.
\end{aligned}
\end{equation}
\end{theorem}

\begin{theorem} Assume
\begin{gather*}
\tau:=2\sqrt{2n}\,a,
\end{gather*}
is fixed as $a\rightarrow0, n\rightarrow\infty$, and let
\begin{gather*}
\sigma(\tau):=\lim_{n\rightarrow\infty}\sigma_n\left(\frac{\tau}{2\sqrt {2n}}\right).
\end{gather*}
Then we obtain the differential equation satisfied by $\sigma(\tau)$
\begin{gather}\label{sigmaODE}
(\tau\sigma '')^2=-4 \left(\sigma-\tau\sigma'-(\sigma')^2\right)\left(\sigma-\tau \sigma '\right),
\end{gather}
which is the celebrated equation (7.104) of JMMS. Moreover, by changing variable $\tau\rightarrow\frac{i}{2}\tau$, $i$ denoting the imaginary unit, \eqref{sigmaODE} converts to the $\sigma$ form of Painlev\'{e} V \cite{JimboMiwa1981} with $\nu_0=\nu_1=\nu_2=\nu_3=0$.
\end{theorem}

We shall be concerned with the behavior of the gap probability for large variable $\tau.$
We can of course, make use of \eqref{sigmaODE} to investigate the asymptotic behavior of the gap probability. However, we found a convenient way motivated by the relation between Hermite and Laguerre polynomials (with special values of the parameter $\alpha$) given by Szeg\"{o} \cite{Szego1939} (formula (5.6.1)).

To begin with, we introduce a change of variable $x^2=t$ in the normalization relation for our orthogonal polynomials
\begin{align*}
h_j(a)=&\int_{-\infty}^{\infty}P_j^2(x;a)w(x,a)dx,\qquad w(x,a)=\theta (x^2-a^2){\rm e}^{-x^2},
\end{align*}
to find
\begin{align*}
h_{2n}(a)=&\int_{-\infty}^{\infty}P_{2n}^2(x;a)\theta (x^2-a^2){\rm e}^{-x^2}dx\\
=&2\int_{0}^{\infty}P_{2n}^2(\sqrt{t};a)\theta (t-a^2)\frac{{\rm e}^{-t}}{2\sqrt{t}}dt\\
=&\int_{a^2}^{\infty}\widetilde{P}_n^2(t;a)t^{-\frac{1}{2}}{\rm e}^{-t}dt=:\widetilde{h}_n(a),\\
\intertext{and}
h_{2n+1}(a)=&\int_{-\infty}^{\infty}P_{2n+1}^2(x;a)\theta (x^2-a^2){\rm e}^{-x^2}dx\\
=&2\int_{0}^{\infty}P_{2n+1}^2(\sqrt{t};a)\theta (t-a^2)\frac{{\rm e}^{-t}}{2\sqrt{t}}dt\\
=&\int_{a^2}^{\infty}\widehat{P}_n^2(t;a)t^{\frac{1}{2}}{\rm e}^{-t}dt=:\widehat{h}_n(a).
\end{align*}
Here $\widetilde{P}_n(t;a)$ and $\widehat{P}_n(t;a)$ are monic polynomials of degree $n$ in the variable $t$, orthogonal with respect to $t^{-\frac{1}{2}}{\rm e}^{-t}$ and $t^{\frac{1}{2}}{\rm e}^{-t}$ over $[a^2,\infty)$ respectively. Note that
\begin{gather*}
P_{2n}(x;a)=x^{2n}+p(2n,a)x^{2n-2}+\cdots+P_{2n}(0;a),
\end{gather*}
and
\begin{align*}
P_{2n+1}(x;a)=&x^{2n+1}+p(2n+1,a)x^{2n-1}+\cdots+{\rm const.}\,x\\
=&x\left(x^{2n}+p(2n+1,a)x^{2n-2}+\cdots+{\rm const.}\right).
\end{align*}

Define the Hankel determinants generated by $t^{-\frac{1}{2}}{\rm e}^{-t}$ and $t^{\frac{1}{2}}{\rm e}^{-t}$, $a^2\leq t<\infty$, by
\begin{align*}
\widetilde{D}_m(a):=&\det \left(\int_{a^2}^{\infty}t^{i+j}t^{-\frac{1}{2}}{\rm e}^{-t}dt\right)_{i,j=0}^{m-1}=\prod_{l=0}^{m-1}\widetilde{h}_l(a),\\
\widehat{D}_m(a):=&\det \left(\int_{a^2}^{\infty}t^{i+j}t^{\frac{1}{2}}{\rm e}^{-t}dt\right)_{i,j=0}^{m-1}=\prod_{l=0}^{m-1}\widehat{h}_l(a),
\end{align*}
respectively, we readily see that
\begin{gather*}
D_n(a)=\prod_{j=0}^{n-1}h_j(a)=
\begin{cases}
\widetilde{D}_{k+1}\,\widehat{D}_{k}, & n=2k+1,\\
\widetilde{D}_{k}\,\widehat{D}_{k}, & n=2k.
\end{cases}
\end{gather*}
Before proceeding any further we describe in the next section the smallest eigenvalue distribution of the Laguerre unitary ensembles.
\\
\section{The Smallest Eigenvalue Distribution of the Laguerre Unitary Ensembles}
In this section we shall be concerned with the Laguerre weight
\begin{gather*}
w(x,\alpha)=x^{\alpha}{\rm e}^{-x},\qquad x\in[0,\infty),\quad\alpha>-1.
\end{gather*}
The probability that all the eigenvalues are greater than $t$ in the finite $n$ Laguerre unitary ensembles, is given by
\begin{align}
\mathbb{P}(t,\alpha,n)=&\frac{\frac{1}{n!}\int_{(t,\infty)^n}{\prod\limits_{1\leq i<j\leq n}(x_{j}-x_{i})}^{2}\prod\limits^{n}_{k=1}w(x_k,\alpha)dx_k}{\frac{1}{n!}\int_{(0,\infty)^n}{\prod\limits_{1\leq i<j\leq n}(x_{j}-x_{i})}^{2}\prod\limits^{n}_{k=1}w(x_k,\alpha)dx_k}\non\\
=&
\frac{\det\left(\int_t^{\infty}x^{i+j}w(x,\alpha)dx\right)_{i,j=0}^{n-1}}
{\det\left(\int_0^{\infty}x^{i+j}w(x,\alpha)dx\right)_{i,j=0}^{n-1}}\non\\
=:&\frac{D_n(t,\alpha)}{D_n(0,\alpha)}.\label{LUEgapHankel}
\end{align}
Note that $D_n(0,\alpha)$ has a closed-form expression and reads (see \cite{Mehta2006}, p. 321)
\begin{gather}\label{Dn0}
D_n(0,\alpha)=\prod_{j=0}^{n-1}\Gamma(j+1)\Gamma(j+\alpha+1)=G(n+1)\cdot\frac{G(n+\alpha+1)}{G(\alpha+1)},
\end{gather}
where $G(\cdot)$ denotes the Barnes-G function, that satisfies the functional relation $G(z+1)=\Gamma(z)G(z),$ with the `initial' condition,
$G(1)=1.$ For a comprehensive exposition on the $G$ and other related functions see \cite{Voros1987}.

Let $h_j(t,\alpha)$ be the square of the $L^2$ norm of monic polynomial $P_j(.;t,\alpha)$ orthogonal with respect to $w(x,\alpha)$ over $[t,\infty]$:
\begin{gather*}
h_j(t,\alpha)\delta_{jk}:=\int_t^{\infty}P_j(x;t,\alpha)P_k(x;t,\alpha)w(x,\alpha)dx,
\end{gather*}
with the monic $P_j(x;t,\alpha)$ reads,
\begin{gather*}
P_j(x;t,\alpha)=x^j+p(j,t,\alpha) x^{j-1}+\cdots+P_j(0;t,\alpha).
\end{gather*}
It is a well-known fact that the Hankel determinant $D_n(t,\alpha)$ can be evaluated as the product of $h_j(t,\alpha)$, i.e.
\begin{gather*}
D_n(t,\alpha)=\prod_{j=0}^{n-1}h_j(t,\alpha).
\end{gather*}

We list here a number of facts about orthogonal polynomials.
\\
From orthogonality, there follows the three-term recurrence relation
\begin{gather*}
z P_{n}(z;t,\alpha)= {P_{n+1}(z;t,\alpha)}+\alpha_{n}(t,\alpha)P_{n}(z;t,\alpha)+\beta_{n}(t,\alpha)P_{n-1}(z;t,\alpha),\qquad n\geq0,
\end{gather*}
subject to the initial conditions $P_0(z;t,\alpha):=1,~\beta_0(t,\alpha)P_{-1}(z;t,\alpha):=0$, where
\begin{gather*}
\alpha_{n}(t,\alpha)=p(n,t,\alpha)-p(n+1,t,\alpha),\qquad\beta_{n}(t,\alpha)=\frac{h_n(t,\alpha)}{h_{n-1}(t,\alpha)}.
\end{gather*}
Moreover, $P_n(z;t,\alpha)$ has the following integral and determinant representations \cite{Szego1939}:
\begin{align*}
P_{n}(z;t,\alpha)=&\frac{1}{D_n(t,\alpha)}
\frac{1}{n!}\int_{(t,\infty)^{n}}\prod\limits_{1\leq i<j\leq n}(x_{j}-x_{i})^2\prod\limits^{n}_{k=1}(z-x_k)w(x_k,\alpha)dx_k\\
=&\frac{\det\left(\int_t^{\infty}x^{i+j}(z-x)x^{\alpha}{\rm e}^{-x}dx\right)_{i,j=0}^{n-1}}
{\det\left(\int_t^{\infty}x^{i+j}x^{\alpha}{\rm e}^{-x}dx\right)_{i,j=0}^{n-1}}.
\end{align*}
From this, we easily find,
\begin{gather}\label{Pn0Dn}
P_{n}(0;t,\alpha)=(-1)^n\frac{D_n(t,\alpha+1)}{D_n(t,\alpha)},
\end{gather}
so that, we find in view of \eqref{LUEgapHankel},
\begin{gather}\label{Pn0probn}
\frac{P_{n}(0;t,\alpha)}{P_n(0;0,\alpha)}=\frac{\mathbb{P}(t,\alpha+1,n)}{\mathbb{P}(t,\alpha,n)},
\end{gather}
and in view of \eqref{Dn0},
\begin{gather*}
P_{n}(0;0,\alpha)=(-1)^n\frac{\Gamma(n+1+\alpha)}{\Gamma(1+\alpha)}.
\end{gather*}

\subsection{The $\sigma$ from of Painlev\'{e} V for finite $n$ and of Painlev\'{e} III}
Since our Hankel determinant can also be written as
\begin{gather*}
D_n(t,\alpha)=\det\left(\int_0^{\infty}x^{i+j}x^{\alpha}{\rm e}^{-x}\theta(x-t)dx\right)_{i,j=0}^{n-1},
\end{gather*}
we see that this is a special case, of the Hankel determinant generated by the discontinuous Laguerre weight
$x^{\alpha}{\rm e}^{-x}(A+B\theta(x-t)),A\geq0, A+B\geq0$, studied in \cite{BasorChen2009}, by putting $A=0$ and $B=1$. The parameters $A$ and $B$ here should not be confused with integration interval in page 4.
\\
\noindent
{Remark} In \cite{BasorChen2009}, the parameters $A$ and $B$ in the weight did not appear in the $\sigma-$form of the Painlev\'e V.

To study the large $n$ behavior of $D_n(t,\alpha)$, we first recall some results in \cite{BasorChen2009}.
\begin{proposition} Define
\begin{gather*}
R_n(t) := \frac{P_n^2(t;t,\alpha)t^{\alpha}{\rm e}^{-t}}{h_n(t,\alpha)},
\end{gather*}
where $P_n(t;t,\alpha)$ is the evaluation of the orthogonal polynomial $P_n(z;t,\alpha)$ at $z=t$. Then $R_n(t)$ satisfies the following second order differential equation
\begin{equation}\label{t0Rn*}
\begin{aligned}
\left(R_n\right)''=&\frac{1}{2}\left(\frac{1}{R_n-1}+\frac{1}{R_n}\right)(R_n')^2-\frac{R_n'}{t}\\
   &+R_n^3+\left(
   \frac{2 n+\alpha +1}{t}-\frac{3}{2}\right)R_n^2-\left(\frac{2 n+\alpha +1}{t}-\frac{1}{2}\right)R_n-\frac{\alpha ^2 }{2 t^2}\frac{R_n}{R_n-1},
\end{aligned}
\end{equation}
 and can be transformed into a particular Painlev\'{e} V \cite{JimboMiwa1981}, namely, $P_V(0,\;-\alpha^2/2,\:2n+1+\alpha,\;-1/2)$, satisfied by $S_n(t):=1-\frac{1}{R_n(t)}$. The appearance of $n$ in the parameter of the $P_V$ indicates that we are studying the finite $n$ problem.
\end{proposition}
\begin{proposition} The quantity
\begin{gather*}
H_n (t) := t\frac{d}{dt}\log \mathbb{P}(t,\alpha,n)
\end{gather*}
satisfies the following Jimbo-Miwa-Okamoto $\sigma$ form \cite{JimboMiwa1981} of Painlev\'{e} V equation:
\begin{gather}\label{t0Hn}
(t H_n'')^2=4(H_n')^2\left(H_n-n(n+\alpha)-t H_n'\right)+\left((2n+\alpha-t)H_n'+H_n\right)^2.
\end{gather}
Furthermore, $H_n$ is expressed in terms of $R_n$ by
\begin{gather*}
H_n=\frac{t^2}{4}\frac{(R_n')^2}{R_n(R_n-1)}+\frac{t^2}{4}R_n\left( 1-R_n\right)-\left(n+\frac{\alpha}{2}\right)t R_n+\frac{\alpha^2}{4}\frac{R_n}{1-R_n}.
\end{gather*}
\end{proposition}

It was pointed out in \cite{BasorChen2009}, by changing variable $t\rightarrow \frac{s}{4n}$ and $H_n(t)\rightarrow\sigma(s)$ in \eqref{t0Hn},
the coefficient of the highest order term in $n$ gives rise to the Jimbo-Miwa-Okamoto $\sigma$ form of Painlev\'{e} III.
So we treat $R_n$ and equation \eqref{t0Rn*}, with $n$ large, in a similar way. Here are the results.
\begin{proposition} Write
\begin{gather*}
R(s):=\lim_{n\rightarrow\infty}R_n\left(\frac{s}{4n}\right),
\end{gather*}
then the following second order differential equation holds
\begin{gather}\label{R*ODE}
R''=\left(\frac{1}{R-1}+\frac{1}{R}\right)\frac{(R')^2}{2}- \frac{R'}{s}+\frac{R\left(R-1\right)}{2s}-\frac{\alpha^2}{2s^2}\cdot\frac{R}{R-1}.
\end{gather}
Let
\begin{gather}\label{Psdef}
\mathbb{P}(s,\alpha):=\lim_{n\rightarrow\infty}\mathbb{P}\left(\frac{s}{4n},\alpha,n\right),
\end{gather}
and
\begin{gather*}
\sigma(s):=s\frac{d}{ds}\log \mathbb{P}(s,\alpha).
\end{gather*}
Then $\sigma(s)$ satisfies the Jimbo-Miwa-Okamoto $\sigma$ form of $P_{III}$ (see (3.13) of \cite{Jimbo1982})
\begin{gather}\label{sigmaeq}
\left(s\sigma''\right)^2 +\sigma'\left(4 \sigma'+1\right) \left(s \sigma'-\sigma\right)-\alpha^2  \left(\sigma'\right)^2=0.
\end{gather}
The quantity $\sigma(s)$ when expressed in terms of $R(s)$ reads,
\begin{gather}\label{sigmaR}
\sigma(s)=\frac{s^2(R')^2}{4R(R-1)}
-\frac{s}{4} R+\frac{\alpha ^2 }{4}\cdot\frac{R}{1-R}\:.
\end{gather}
Hence, $\mathbb{P}(s,\alpha)$ has the following integral representation
\begin{equation}\label{DR*}
\begin{aligned}
\log \mathbb{P}(s,\alpha)
&=-\int_s^{\infty} \frac{\sigma(\xi)}{\xi} d\xi\\
&=-\int_s^{\infty} \left(\frac{\xi^2(R')^2}{4R(R-1)}
-\frac{\xi }{4} R+\frac{\alpha ^2 }{4}\cdot\frac{R}{1-R}\right)\frac{d\xi}{\xi}.
\end{aligned}
\end{equation}
\end{proposition}
We shall make use of \eqref{R*ODE} to derive the series expansion of $R(s)$ as $s\rightarrow\infty$, and then apply \eqref{DR*} to obtain $\mathbb{P}(s,\alpha)$ for large $s.$
The lemma below gives the bounds for $R(s)$, which we will see is important for later development.
\begin{lemma} $R_n(t)$ and $R(s)$ are bounded by
\begin{gather*}
0\leq R_n(t)<1,\\
0\leq R(s)<1.
\end{gather*}
\end{lemma}
\begin{proof}
Noting that
\begin{gather*}
R_n(t):=\frac{P_n^2(t;t,\alpha)t^{\alpha}{\rm e}^{-t}}{h_n(t,\alpha)}\geq0,\\
\intertext{and}
0<\frac{1}{h_n(t,\alpha)}\int_t^{\infty}P_n^2(y;t,\alpha)\alpha y^{\alpha-1}{\rm e}^{-y}dy=1-R_n(t),
\end{gather*}
we find
\begin{gather*}
0\leq R_n(t)<1,
\end{gather*}
so that
\begin{gather*}
0\leq R(s)=\lim_{n\rightarrow\infty}R_n\left(\frac{s}{4n}\right)<1.
\end{gather*}
\end{proof}

To continue, we obtain, neglecting the derivatives in \eqref{R*ODE} and replacing $R(s)$ by $\widetilde{R}(s)$, we obtain
\begin{gather*}
\left(\widetilde{R}-1\right)^2=\frac{\alpha^2}{s},
\end{gather*}
which has solutions
\begin{gather*}
\widetilde{R}=1\pm\frac{\alpha }{\sqrt{s}}.
\end{gather*}
So it seems reasonable to assume $R(s)$ has an expansion of the form
\begin{gather*}
R(s)=\sum_{j=0}^{\infty}a_j s^{-\frac{j}{2}},\qquad s\rightarrow\infty.
\end{gather*}
Substituting the above into \eqref{R*ODE}, followed by comparing the corresponding coefficients on both sides, we find $a_0=1$ and $a_1=\pm\alpha$.
Since $R(s)<1$, we choose $a_1=-\alpha$. By direct computations, we eventually arrive at the following expansion formula for $R(s)$.
\begin{theorem} The expansion holds,
\begin{equation}\label{R*t0}
\begin{aligned}
R(s)=&1-\frac{\alpha }{\sqrt{s}}-\frac{\alpha }{8}s^{-\frac{3}{2}}-\frac{\alpha
   ^2}{4}s^{-2}-\left(\frac{3
   \alpha ^3}{8}+\frac{27
   \alpha
   }{128}\right)s^{-\frac{5}{2}}\\
&-\left(\frac{\alpha ^4}{2}+\frac{9\alpha^2}{8}\right)s^{-3}-\left(\frac{5 \alpha
   ^5}{8}+\frac{225 \alpha^3}{64}+\frac{1125 \alpha}{1024}\right)s^{-\frac{7}{2}}\\
&-\left(\frac{3
   \alpha ^6}{4}+\frac{135
   \alpha ^4}{16}+\frac{81
   \alpha
   ^2}{8}\right)s^{-4}+O\left(s^{-\frac{9}{2}}\right),\qquad s\rightarrow\infty.
\end{aligned}
\end{equation}
Thus, from \eqref{sigmaR}, it follows that
\begin{equation}\label{HT0}
\begin{aligned}
\sigma(s)=&-\frac{s}{4}+\frac{\alpha}{2}\sqrt{s}-\frac{\alpha
   ^2}{4}-\frac{\alpha
   }{16\sqrt{s} }-\frac{\alpha ^2}{16}s^{-1}-\left(\frac{\alpha^3}{16}+\frac{9\alpha}{256}\right)s^{-\frac{3}{2}}\\
&-\left(\frac{\alpha^4}{16}+\frac{9\alpha^2}{64}\right)s^{-2}-\left(\frac{\alpha^5}{16}+\frac{45\alpha^3}{128}+\frac{225\alpha}{2048}\right)s^{-\frac{5}{2}}\\
&
-\left(\frac{\alpha^6}{16}+\frac{45\alpha^4}{64}+\frac{27\alpha^2}{32}\right)s^{-3}+O\left(s^{-\frac{7}{2}}\right),\qquad s\rightarrow\infty.
\end{aligned}
\end{equation}
\end{theorem}

\begin{remark}
Our asymptotic expansion \eqref{HT0} coincides with (3.1) in \cite{TracyWidom1994}, since $-\sigma(s)=\sigma(s;1)$. In fact, we observe that
$\mathbb{P}(s,\alpha)=D(J;1),$
where $D(J;1)$ denotes the Fredholm determinant of the operator with Bessel kernel.
As a result, by definition, we find $-\sigma(s)=-s\frac{d}{ds}\log \mathbb{P}(s,\alpha)=-s\frac{d}{ds}\log D(J;1)=\sigma(s;1)$.
See \cite{TracyWidom1994} for more details.
\end{remark}

Finally, according to \eqref{DR*}--\eqref{HT0},
we obtain the the asymptotic expansion for $\mathbb{P}(s,\alpha)$.
\begin{theorem} As $s\rightarrow\infty$, we have
\begin{equation}\label{probs}
\begin{aligned}
\log\mathbb{P}(s,\alpha)=&c_1(\alpha)-\frac{s}{4}+\alpha \sqrt{s}-\frac{\alpha^2}{4}\log s+\frac{\alpha }{8}s^{-\frac{1}{2}}+\frac{\alpha
   ^2}{16}s^{-1}
   +\left(\frac{\alpha^3}{24}+\frac{3\alpha}{128}\right)s^{-\frac{3}{2}}\\
&+\left(\frac{\alpha^4}{32}+\frac{9\alpha^2}{128}\right)s^{-2}
+\left(\frac{\alpha^5}{40}+\frac{9\alpha^3}{64}+\frac{45\alpha}{1024}\right)s^{-\frac{5}{2}}\\
&+\left(\frac{\alpha^6}{48}+\frac{15\alpha^4}{64}+\frac{9\alpha^2}{32}\right)s^{-3}+O\left(s^{-\frac{7}{2}}\right),
\end{aligned}
\end{equation}
where $c_1(\alpha),$ the constant of integration, is shown later to be
\begin{gather*}
c_1(\alpha)=\log \frac{G(\alpha+1)}{(2\pi)^{\alpha/2}}.
\end{gather*}
\end{theorem}

We find from \eqref{Pn0probn} and \eqref{probs} that
\begin{equation}\label{c1Pn0}
\begin{aligned}
\lim_{n\rightarrow\infty}&\frac{P_{n}(0;\frac{s}{4n},\alpha)}{P_n(0;0,\alpha)}=\frac{\mathbb{P}(s,\alpha+1)}{\mathbb{P}(s,\alpha)}\\
&\sim\exp\left[c_1(\alpha+1)-c_1(\alpha)+\sqrt{s}-\left(\frac{\alpha}{2}+\frac{1}{4}\right)\log s+O\left(s^{-\frac{1}{2}}\right)\right],
\qquad s\rightarrow\infty.
\end{aligned}
\end{equation}
We shall apply Dyson's Coulomb fluid \cite{Dyson1962JMP} to derive the asymptotic formula for $P_n(0;t,\alpha)$.
The result combined with \eqref{c1Pn0} enables us to find the constant $c_1(\alpha)$.

\subsection{The evaluation of $P_n(0;t,\alpha)$ via Dyson's Coulomb fluid}
Recall that the equilibrium density $\rho(x)$ is given by \eqref{densitycon.} with $a=t$:
\begin{gather*}
\rho(x)=\frac{1}{2\pi^2}\sqrt{\frac{b-x}{x-t}}\,P\int_t^b\frac{dy}{y-x}\sqrt{\frac{y-t}{b-y}}v'(y),\qquad x\in(t,b).
\end{gather*}
Here
\begin{gather*}
v(x)=-\log w(x,\alpha)=-\alpha\log x+x.
\end{gather*}
We find with the integral identities in the Apoendix,
\begin{gather}\label{LUEsmallestdensity}
\rho(x)=\frac{1}{2\pi}\sqrt{\frac{b-x}{x-t}}\left(1-\frac{\alpha}{x}\sqrt{\frac{t}{b}}\right).
\end{gather}
Note that condition $\sqrt{tb}>\alpha$ and that $$\frac{d}{dx}\left(\rho(x)\sqrt{\frac{x-t}{b-x}}\right)>0,$$ show that $\rho
(x)>0$ for $t<x<b$.
Hence, from the normalization condition $\int_t^b\rho(x)dx=n$, it follows that
\begin{gather}\label{LUEbcon.}
n=\frac{b-t}{4}+\frac{\alpha}{2}\left(\sqrt{\frac{t}{b}}-1\right),
\end{gather}
from which we obtain a cubic equation satisfied by $b$
\begin{gather}\label{LUEbeq}
b\big(b-2\left(2n+\alpha\right)-t\big)^2=4\alpha^2t.
\end{gather}
Equations \eqref{LUEsmallestdensity}-\eqref{LUEbeq} were derived in \cite{ChenManning1996} (equations (15)-(17) with $\beta=2$).



Since the recurrence coefficient $\alpha_n$, is asymptotic to the centre of mass of the support $[t,b]$, namely $\frac{b+t}{2}$ for large $n$, see for
example \cite{ChenIsmail1997JPA:MG},
we find by using the equality $\alpha_n=2n+1+\alpha+tR_n$ (see (3.9), \cite{BasorChen2009}) that
\begin{align*}
b\sim&2\alpha_n-t=2(2n+1+\alpha+tR_n)-t\\
\sim&4n+2\alpha+t+2t(R_n-1)\\
<&4n+2\alpha+t,
\end{align*}
where the last inequality is due to $R_n<1$.

Now we are in a position to derive the large $n$ expansion for $b$.
\begin{lemma}\label{bt0-expansion}
Provided $t>0$, we have for large $n$
\begin{equation}\label{bt}
\begin{aligned}
b=&4 n+2 \alpha+t-\alpha  t^{\frac{1}{2}}n^{-\frac{1}{2}}+\frac{1}{8}\left(\alpha  t^{\frac{3}{2}}+2 \alpha
   ^2 t^{\frac{1}{2}}\right)n^{-\frac{3}{2}}\\
&-\frac{1}{8}\alpha ^2 t n^{-2}-\frac{3}{128}\left(\alpha  t^{\frac{5}{2}}+4 \alpha ^2 t^{\frac{3}{2}}+4 \alpha ^3t^{\frac{1}{2}}\right)n^{-\frac{5}{2}}\\
&+\frac{1}{16}\left(\alpha ^2 t^2+2 \alpha ^3 t\right)n^{-3}+O\left(n^{-\frac{7}{2}}\right).
\end{aligned}
\end{equation}
Moreover, putting $t=\frac{s}{4n}$ which tends to $0$ as $n\rightarrow\infty$, we find
\begin{equation}\label{bs}
\begin{aligned}
b=&4n+2\alpha+\left(\frac{s}{4}-\frac{\alpha
}{2}\sqrt{s}\right)n^{-1}+\frac{\alpha ^2}{8}\sqrt{s}n^{-2}\\
&+\left(\frac{\alpha}{64} s^{\frac{3}{2}}-\frac{\alpha ^2}{32} s-\frac{3\alpha ^3}{64} \sqrt{s}\right)n^{-3}+O\left(n^{-4}\right).
\end{aligned}
\end{equation}
\end{lemma}
\begin{proof}
Given $t>0$, we find from \eqref{LUEbeq} that
\begin{gather*}
b\sim 4n+2\alpha+t+O(n^{-\tfrac{1}{2}}),\qquad n\rightarrow\infty.
\end{gather*}
Hence, we assume the following expansion,
\begin{gather*}
b=4n+2\alpha+t+\sum_{j=1}^{\infty}a_j n^{-\tfrac{j}{2}}.
\end{gather*}
Substituting this into \eqref{LUEbeq}, by comparing the corresponding coefficients on both sides, we get $a_1=\pm\alpha\sqrt{t}$.
Since $b<4n+2\alpha+t$, we have to choose $a_1=-\alpha\sqrt{t}$.
By simple calculations, we obtain \eqref{bt}, from which, by putting $t=\frac{s}{4n}$ and sending $n$ to $\infty$, \eqref{bs} follows.
\end{proof}

According to \eqref{PnS1S2general}, for large $n$, our monic orthogonal polynomials are approximated by
\begin{align*}
P_n(0;t,\alpha)&\sim {\rm e}^{-S_1(0;t,\alpha)-S_2(0;t,\alpha)},
\end{align*}
where
\begin{align*}
{\rm e}^{-S_1(0;t,\alpha)}=&\frac{1}{2}\left[\left(\frac{b}{t}\right)^{\frac{1}{4}}
+\left(\frac{t}{b}\right)^{\frac{1}{4}}\right],
\end{align*}
and
\begin{equation*}
\begin{aligned}
S_2(0;t,\alpha)=&-n\log\left(\frac{\sqrt{-t}+\sqrt{-b}}{2}\right)^2\\
&+\frac{1}{2\pi}\int_t^b \frac{-\alpha\log x+x}{\sqrt{(b-x)(x-t)}}\left[-\frac{\sqrt{tb}}{x}+1\right]dx.
\end{aligned}
\end{equation*}
With the aid of the integral identities listed in the Appendix, and choosing by the branch $-t=t{\rm e}^{\pi i}$ and $-b=b {\rm e}^{\pi i}$, we obtain
\begin{equation*}
\begin{aligned}
{\rm e}^{-S_2(0;t,\alpha)}=&(-1)^n 4^{-n-\alpha}\left(\sqrt{b}+\sqrt{t}\right)^{2n}\left[\left(\tfrac{b}{t}\right)^{\frac{1}{4}}
+\left(\tfrac{t}{b}\right)^{\frac{1}{4}}\right]^{2\alpha}{\rm e}^{-\left(\frac{\sqrt{b}-\sqrt{t}}{2}\right)^2}.
\end{aligned}
\end{equation*}

Finally, by using \eqref{bt}, we give an evaluation of $P_n(0;t,\alpha)$ for large $n$ and thus determine the constant $c_1(\alpha)$.
\begin{theorem} The monic polynomial $P_n(x;t,\alpha)$ orthogonal with respect to $x^{\alpha}{\rm e}^{-x},\alpha>-1$, over $[t,\infty)$ is approximated at $x=0$ by
\begin{align*}
(-1)^nP_n(0;t,\alpha)\sim&\left(4t\right)^{-\frac{\alpha}{2}-\frac{1}{4}}n^{n+\frac{\alpha}{2}+\frac{1}{4}}
{\rm e}^{-n+\sqrt{4nt}},\qquad n\rightarrow\infty.
\end{align*}
Hence, under the assumption that $t\rightarrow0$ and $n\rightarrow\infty$ such that $s=4nt$ is fixed, we find
\begin{gather}\label{Pn0s}
(-1)^nP_n(0;\tfrac{s}{4n},\alpha)\sim s^{-\frac{\alpha}{2}-\frac{1}{4}}n^{n+\alpha+\frac{1}{2}}
{\rm e}^{-n+\sqrt{s}},\qquad n\rightarrow\infty.
\end{gather}
\end{theorem}

\begin{corollary} The constant $c_1(\alpha)$, appearing in \eqref{probs}, is  identified to be
\begin{gather*}
c_1(\alpha)=\log \frac{G(\alpha+1)}{(2\pi)^{\alpha/2}}
\end{gather*}
where $G(\cdot)$ is the Barns-G function.
\end{corollary}
\begin{proof} In what follows, the symbol $\sim$ refers to 'asymptotic to' for large $n$.

Combining Stirling's formula \cite{Lebedev}
\begin{gather*}
n!\sim\sqrt{2\pi n}\left(\frac{n}{{\rm e}}\right)^n
\end{gather*}
and the standard asymptotic approximation for Gamma function
\begin{gather*}
\Gamma(n+\alpha)\sim\Gamma(n)n^{\alpha},\qquad \alpha\in\mathbb{C},
\end{gather*}
we get
\begin{gather*}
\left(\frac{n}{{\rm e}}\right)^n\sim\frac{\Gamma(n+1+\alpha)}{\sqrt{2\pi}}n^{-\alpha-\frac{1}{2}}.
\end{gather*}
or equivalently,
\begin{gather*}
n^{n+\alpha+1/2}\sim \frac{\Gamma(n+1+\alpha)}{\sqrt{2\pi}}\:{\rm e}^{n}.
\end{gather*}
Hence, it follows from \eqref{Pn0s} that
\begin{align*}
(-1)^nP_n(0;\tfrac{s}{4n},\alpha)\sim&
\frac{\Gamma(n+1+\alpha)}{\sqrt{2\pi}}s^{-\frac{\alpha}{2}-\frac{1}{4}}
{\rm e}^{\sqrt{s}}.
\end{align*}
Therefore, noting the explicit evaluation of the monic Laguerre polynomials at $0,$
\begin{align*}
(-1)^nP_{n}(0;0,\alpha)=\frac{\Gamma(n+1+\alpha)}{\Gamma(1+\alpha)},
\end{align*}
we finally obtain
\begin{align*}
\frac{P_n(0;\tfrac{s}{4n},\alpha)}{P_n(0;0,\alpha)}\sim&\frac{\Gamma(1+\alpha)}{\sqrt{2\pi}}s^{-\frac{\alpha}{2}-\frac{1}{4}}
{\rm e}^{\sqrt{s}}\\
=&\exp\left[\log\left(\frac{\Gamma(1+\alpha)}{\sqrt{2\pi}}\right)+\sqrt{s}-\left(\frac{\alpha}{2}+\frac{1}{4}\right)\log s\right].
\end{align*}
Comparing this with \eqref{c1Pn0} yields
\begin{gather*}
c_1(\alpha+1)-c_1(\alpha)=\log\left(\frac{\Gamma(1+\alpha)}{\sqrt{2\pi}}\right),
\end{gather*}
which, according to the properties of the Barnes-G function, leads to our conclusion.
\end{proof}

\section{The Asymptotics of the Gap Probability Distribution of the Gaussian Unitary Ensembles}
The gap probability distribution of GUE on $(-\tfrac{b}{\sqrt{2n}},\tfrac{b}{\sqrt{2n}})$ with $n$ large enough is described by Ehrhardt \cite{Ehrhardt2006}:
\begin{align*}
\mathbb{P}(b)=\det (I-K_{2b}),
\end{align*}
where $\det (I-K_{2b})$ is the Fredholm determinant with $K_{2b}$ having the kernel $\frac{\sin(x-y)}{\pi (x-y)}\chi_{(-b,b)}(y)$. For this reason, we shall make use of the asymptotic expression \eqref{probs} for large $n$ Hankel determinant to deal with our problem
under the double scaling
\begin{gather*}
b:=\sqrt{2n}\,a=\frac{\tau}{2}.
\end{gather*}
Note that $\tau=2\sqrt{2n}\,a$.
\begin{theorem} The probability that the interval $(-\tfrac{b}{\sqrt{2n}},\tfrac{b}{\sqrt{2n}}), n\rightarrow\infty$, is free of eigenvalues in the Gaussian unitary ensembles is approximated by
\begin{equation}\label{limitprobG}
\begin{aligned}
\log\mathbb{P}(b)=&-\frac{b^2}{2}-\frac{\log b}{4}+\frac{\log 2}{12}+3\zeta'(-1)\\
&+\frac{1}{32\,b^{2}}+\frac{5}{128\,b^{4}}+\frac{131}{768\,b^{6}}+O\left(b^{-8}\right), \qquad b\rightarrow\infty.
\end{aligned}
\end{equation}
\end{theorem}
\begin{proof}
Recall that the gap probability of GUE on $(-a,a)$ is given by
\begin{gather*}
\mathbb{P}(a,n)=\frac{D_n(a)}{D_n(0)},
\end{gather*}
and the Hankel determinant $D_n(a)$ is connected with the smallest eigenvalue distribution of LUE by the relation
\begin{gather*}
D_n(a)=
\begin{cases}
\widetilde{D}_{k+1}\,\widehat{D}_{k}, & n=2k+1,\\
\widetilde{D}_{k}\,\widehat{D}_{k}, & n=2k.
\end{cases}
\end{gather*}
Here
\begin{align*}
\widetilde{D}_m(a):=&\det \left(\int_{a^2}^{\infty}t^{i+j}t^{-\frac{1}{2}}{\rm e}^{-t}dt\right)_{i,j=0}^{m-1},\\
\widehat{D}_m(a):=&\det \left(\int_{a^2}^{\infty}t^{i+j}t^{\frac{1}{2}}{\rm e}^{-t}dt\right)_{i,j=0}^{m-1}.
\end{align*}
From this, we find
\begin{align*}
\mathbb{P}(b):=&\lim_{n\rightarrow\infty}\mathbb{P}\left(a,n\right)=\lim_{n\rightarrow\infty}\frac{D_n\left(a\right)}{D_n(0)}\\
=&\lim_{k\rightarrow\infty}\frac{\widetilde{D}_{k}\left(a\right)}{\widetilde{D}_{k}(0)}
\cdot\lim_{k\rightarrow\infty}\frac{\widehat{D}_{k}\left(a\right)}{\widehat{D}_{k}(0)},
\end{align*}
so that, in view of \eqref{LUEgapHankel} and noting that $b^2\sim 4k\,a^2$ as $k\rightarrow\infty$,
\begin{align*}
\mathbb{P}(b)=&\lim_{k\rightarrow\infty}\mathbb{P}(a^2,-\tfrac{1}{2},k)\cdot\lim_{k\rightarrow\infty}\mathbb{P}(a^2,\tfrac{1}{2},k)\\
=&\mathbb{P}\left(b^2,-\tfrac{1}{2}\right)\cdot\mathbb{P}\left(b^2,\tfrac{1}{2}\right).
\end{align*}
Here $\mathbb{P}(a^2,\alpha,k)$ is the probability that all the eigenvalues of $k\times k$ Hermitian matrices with weight $x^{\alpha}{\rm e}^{-x}$ are greater than $a^2$ and $\mathbb{P}\left(b^2,\alpha\right)$ defined by \eqref{Psdef} is the scaled limiting probability of $\mathbb{P}(a^2,\alpha,k)$, i.e. $\mathbb{P}\left(b^2,\alpha\right)=\lim_{k\rightarrow\infty}\mathbb{P}\left(\tfrac{b^2}{4k},\alpha,k\right)$.

Therefore, according to \eqref{probs}, we see that
\begin{align*}
\log\mathbb{P}(b)=&\log\mathbb{P}\left(b^2,-\tfrac{1}{2}\right)+\log\mathbb{P}\left(b^2,\tfrac{1}{2}\right)\\
=&c_1\left(-\tfrac{1}{2}\right)+c_1(\tfrac{1}{2})-\frac{b^2}{2}-\frac{\log b }{4}+\frac{1}{32\,b^2}\\
&+\frac{5}{128\,b^4}
+\frac{131}{768\,b^6}+O\left(b^{-8}\right),\qquad b\rightarrow\infty,
\end{align*}
where the constant term reads
\begin{gather*}
c_1\left(-\tfrac{1}{2}\right)+c_1(\tfrac{1}{2})=\log \left[G(\tfrac{1}{2})G(\tfrac{3}{2})\right]=\log\left[G(\tfrac{1}{2})^2\Gamma(\tfrac{1}{2})\right].
\end{gather*}
taking note that $G(\tfrac{3}{2})=\Gamma(\tfrac{1}{2})G(\tfrac{1}{2}).$
\\
\noindent
It follows from eq.(6.39) of Voros \cite{Voros1987}, that,
\begin{gather*}
G(\tfrac{1}{2})={\rm e}^{3\zeta'(-1)/2}\pi^{-1/4}2^{1/24},
\end{gather*}
where $\zeta(\cdot)$ is the Riemann zeta function, and lead to
\begin{gather*}
c_1\left(-\tfrac{1}{2}\right)+c_1(\tfrac{1}{2})=\frac{\log 2}{12}+3\zeta'(-1).
\end{gather*}
The proof is completed.
\end{proof}
\begin{remark}
Our asymptotic formula \eqref{limitprobG} coincides with the one found by Ehrhardt \cite{Ehrhardt2006}.
\end{remark}
\section{The Smallest Eigenvalue Distribution of the Jacobi Unitary Ensembles}

Let $\mathbb{P}(t,\alpha,\beta,n)$ denote the probability that all the eigenvalues of the Jacobi unitary ensembles with the weight $x^{\alpha}(1-x)^{\beta},~x\in[0,1]$, are between $t$ and 1, where $t>0.$
It was shown in \cite{ChenZhang2010} that
\begin{gather*}
\sigma_n(t):=t(t-1)\frac{d}{dt}\log \mathbb{P}(t,\alpha,\beta,n)+d_1t+d_2
\end{gather*}
satisfies the Jimbo-Miwa-Okamoto $\sigma$ form of Painlev\'{e} VI. Here $d_1$ and $d_2$ are constants depending on $n,\alpha$ and $\beta$.
Under the assumption that $t\rightarrow0$ and $n\rightarrow\infty$ such that $s=4n^2t$ is fixed,
we prove that the $\sigma$ form of Painlev\'{e} VI is reduced to the $\sigma$ form of Painlev\'{e} of a particular III satisfied by $s\frac{d}{ds}\log\mathbb{P}(\frac{s}{4n^2},\alpha,\beta,n)$.
Thus, once again we can derive the asymptotic expansion for the double-scaled probability.
\\
\noindent
{\bf Remark} It will be shown later that the parameter $\beta$  does not appear in this Painleve III.
\\
\noindent
{\bf Remark} We shall see later that,
 \begin{align*}
\mathbb{P}\left(\frac{s}{4n^2},\alpha,\beta,n\right)\sim&\exp\left[c_2(\alpha,\beta)-\frac{s}{4}+\alpha \sqrt{s}-\frac{\alpha^2}{4}
\log s+O\left(s^{-\frac{1}{2}}\right)\right].
\end{align*}
\noindent
{\bf Remark} The constant reads.
\begin{gather*}
c_2(\alpha,\beta)=\log \left[\frac{G(\alpha+1)G^2(\beta+1)}{(2\pi)^{(\alpha+\beta)/2}}\right]
+\frac{\beta(\beta-1)}{2}-(\beta+\tfrac{1}{2})\log\Gamma(\beta),\;\beta>0
\end{gather*}
whose determination is based on the evaluation, at $z=0$ and $z=1$, of the monic polynomials $P_n(z;t,\alpha,\beta)$ orthogonal with respect to $x^{\alpha}(1-x)^{\beta}$ over the interval $[t,1]$. \\
\noindent
The constant $c_2(\alpha,\beta)$ appears to be new.

Consider the Jacobi weight
\begin{gather*}
w(x,\alpha,\beta):=x^{\alpha}(1-x)^{\beta},\qquad x\in[0,1],\quad\alpha,\beta>0.
\end{gather*}

As usual, the probability that all the eigenvalues are between $t$ and 1,  in the Jacobi unitary ensembles reads
\begin{align}
\mathbb{P}(t,\alpha,\beta,n)&=\frac{\frac{1}{n!}\int_{(t,1)^n}{\prod\limits_{1\leq i<j\leq n}(x_{j}-x_{i})}^{2}\prod\limits^{n}_{k=1}w(x_k,\alpha,\beta)dx_k}{\frac{1}{n!}\int_{(0,1)^n}{\prod\limits_{1\leq i<j\leq n}(x_{j}-x_{i})}^{2}\prod\limits^{n}_{k=1}w(x_k,\alpha,\beta)dx_k}\non\\
&=\frac{\det\left(\int_t^{1}x^{i+j}w(x,\alpha,\beta)dx\right)_{i,j=0}^{n-1}
}{\det\left(\int_0^{1}x^{i+j}w(x,\alpha,\beta)dx\right)_{i,j=0}^{n-1}}=:\frac{D_n(t,\alpha,\beta)}{D_n(0,\alpha,\beta)},\label{PnDnJ}
\end{align}
where $D_n(0,\alpha,\beta)$ is given by (\cite{Mehta2006}, p. 310)
\begin{align}
D_n(0,\alpha,\beta)=&\prod_{j=0}^{n-1}\frac{j!\Gamma(j+\alpha+1)\Gamma(j+\beta+1)}{\Gamma(n+j+\alpha+\beta+1)}\label{Dn0J}\\
=&G(n+1)\frac{G(n+\alpha+1)G(n+\beta+1)G(n+\alpha+\beta+1)}{G(\alpha+1)G(\beta+1)G(2n+\alpha+\beta+1)}\non\\
=&D_n(0,\beta,\alpha).\non
\end{align}

Let $h_j(t,\alpha,\beta)$ be the square of the $L^2$ norm of the monic polynomials orthogonal with respect to $w(x,\alpha,\beta)$ over $[t,1]$, i.e.
\begin{gather*}
h_j(t,\alpha,\beta)\delta_{jk}:=\int_t^{1}P_j(x;t,\alpha,\beta)P_k(x;t,\alpha,\beta)w(x,\alpha,\beta)dx.
\end{gather*}
Note the monomial expansion for $P_j$:
\begin{gather*}
P_j(x;t,\alpha,\beta)=x^j+p(j,t,\alpha,\beta) x^{j-1}+\cdots+P_j(0;t,\alpha,\beta),
\end{gather*}
and
\begin{gather*}
D_n(t,\alpha,\beta)=\prod_{j=0}^{n-1}h_j(t,\alpha,\beta).
\end{gather*}
Again the following integral and determinant representations for the orthogonal polynomials hold:
\begin{align*}
P_{n}(z;t,\alpha,\beta)=&\frac{1}{D_n(t,\alpha,\beta)}
\frac{1}{n!}\int_{(t,1)^{n}}\prod\limits_{1\leq i<j\leq n}(x_{j}-x_{i})^2\prod\limits^{n}_{k=1}(z-x_k)w(x_k,\alpha,\beta)dx_k\\
=&\frac{\det\left(\int_t^{1}x^{i+j}(z-x)x^{\alpha}(1-x)^{\beta}dx\right)_{i,j=0}^{n-1}}{\det\left(\int_t^{1}x^{i+j}x^{\alpha}(1-x)^{\beta}dx\right)_{i,j=0}^{n-1}},
\end{align*}
from which we find
\begin{gather*}
P_{n}(0;t,\alpha,\beta)=(-1)^n\frac{D_n(t,\alpha+1,\beta)}{D_n(t,\alpha,\beta)},
\end{gather*}
and
\begin{gather*}
P_{n}(1;t,\alpha,\beta)=\frac{D_n(t,\alpha,\beta+1)}{D_n(t,\alpha,\beta)}.
\end{gather*}
\noindent
{\bf Remark} The equation above will be instrumental for the determination of $c_2(\alpha,\beta).$
As a consequence, we find, according to \eqref{PnDnJ}
\begin{gather}\label{Pn0DnJ}
\frac{P_{n}(0;t,\alpha,\beta)}{P_n(0;0,\alpha,\beta)}=\frac{\mathbb{P}(t,\alpha+1,\beta,n)}{\mathbb{P}(t,\alpha,\beta,n)},
\end{gather}
and
\begin{gather}\label{Pn1DnJ}
\frac{P_{n}(1;t,\alpha,\beta)}{P_n(1;0,\alpha,\beta)}=\frac{\mathbb{P}(t,\alpha,\beta+1,n)}{\mathbb{P}(t,\alpha,\beta,n)}.
\end{gather}
Moreover, by \eqref{Dn0J}, we obtain the explicit formulas
\begin{align}
(-1)^nP_{n}(0;0,\alpha,\beta)=&\frac{D_n(0,\alpha+1,\beta)}{D_n(0,\alpha,\beta)}\nonumber\\
=&\frac{\Gamma(n+\alpha+\beta+1)\Gamma(n+\alpha+1)}{\Gamma(2n+\alpha+\beta+1)\Gamma(\alpha+1)},\label{Pn00J}
\end{align}
and
\begin{align}
P_n(1;0,\alpha,\beta)=&\frac{D_n(0,\alpha,\beta+1)}{D_n(0,\alpha,\beta)}\nonumber\\
=&\frac{\Gamma(n+\alpha+\beta+1)\Gamma(n+\beta+1)}{\Gamma(2n+\alpha+\beta+1)\Gamma(\beta+1)}\nonumber\\
=&(-1)^nP_{n}(0;0,\beta,\alpha)\label{Pn10J}.
\end{align}

\subsection{The $\sigma$ form of Painlev\'{e} VI for finite $n$ and of Panilev\'{e} III \\
 under double scaling}
The following Hankel determinant was studied in \cite{ChenZhang2010}
\begin{gather*}
D_n(t,\alpha,\beta)=\det\left(\int_0^{1}x^{i+j}x^{\alpha}(1-x)^{\beta}\{A+B\,\theta(x-t)\}dx\right)_{i,j=0}^{n-1},
\end{gather*}
where $A\geq0, A+B\geq0$. The special case $A=0, B=1$ reads
\begin{align*}
D_n(t,\alpha,\beta)=&\det\left(\int_0^{1}x^{i+j}x^{\alpha}(1-x)^{\beta}\theta(x-t)dx\right)_{i,j=0}^{n-1}\\
=&\det\left(\int_t^{1}x^{i+j}x^{\alpha}(1-x)^{\beta}dx\right)_{i,j=0}^{n-1},
\end{align*}
the Hankel determinant that concerns us.

Hence, in order to carry out large $n$ analysis of $\mathbb{P}(t,\alpha,\beta,n)$ , we first restate the main result in \cite{ChenZhang2010}.
Here we use the symbol $\sigma_n(t)$ instead of $\sigma(t)$ for the convenience of later discussion.

\begin{proposition}\label{sigmaneqJ}
The quantity
\begin{gather}\label{sigmandefJ}
\sigma_n(t):=t(t-1)\frac{d}{dt}\log D_n(t,\alpha,\beta)+d_1t+d_2,
\end{gather}
with
\begin{align*}
d_1=&-\frac{(2n+\alpha+\beta)^2}{4},\\
d_2=&\frac{1}{4}\left(2n(n+\alpha+\beta)+\beta(\alpha+\beta)\right),
\end{align*}
satisfies the following Jimbo-Miwa-Okamoto \cite{JimboMiwa1981} $\sigma$ form of Painlev\'{e} VI
\begin{equation}\label{sigmanPviJ}
\begin{aligned}
\sigma_n'\{t(t-1)\sigma_n''\}^2+\{2\sigma_n'(t\sigma_n'-\sigma_n)-(\sigma_n')^2-\nu_1\nu_2\nu_3\nu_4\}^2\\
=(\sigma_n'+\nu_1^2)(\sigma_n'+\nu_2^2)(\sigma_n'+\nu_3^2)(\sigma_n'+\nu_4^2),
\end{aligned}
\end{equation}
with parameters
\begin{gather*}
\nu_1=\frac{\alpha+\beta}{2},\qquad\nu_2=\frac{\beta-\alpha}{2},\qquad\nu_3=\nu_4=\frac{2n+\alpha+\beta}{2},
\end{gather*}
and the initial conditions
\begin{gather*}
\sigma_n(0)=d_2,\qquad \sigma_n'(0)=d_1.
\end{gather*}
\end{proposition}
\noindent
{\bf Remark} The parameters $A$ and $B$ did not appear in the equation satisfied by the $\sigma-$function, \cite{ChenZhang2010}.
\\
\noindent
We see that the weight $x^{\alpha}(1-x)^{\beta},~x\in[0,1]$, can be converted to the standard Jacobi weight $(1-x)^{\alpha}(1+x)^{\beta},~x\in[-1,1]$,  by a change of variable $x\rightarrow\frac{1-x}{2}$. The Jacobi polynomials $P_n^{(\alpha,\beta)}(x)$ orthogonal with respect to $(1-x)^{\alpha}(1+x)^{\beta},~x\in[-1,1]$. Furthermore, the following property (see \cite{Szego1939}, Theorem 8.21.12) holds
\begin{gather*}
\lim_{n\rightarrow\infty} n^{-\alpha}P_n^{(\alpha,\beta)}\left(1-\frac{z^2}{2n^2}\right)=(z/2)^{-\alpha}J_{\alpha}(z),
\end{gather*}
where $J_{\alpha}(\cdot)$ is the Bessel function of the first kind of order $\alpha$. This motivates us to execute the double scaling,
namely, $t\rightarrow 0$ and $n\rightarrow\infty$, such that $s=4n^2t$ is fixed.
With $t=\frac{s}{4n^2}$, we obtain the following result from Proposition \ref{sigmaneqJ}.

\begin{theorem}
Let
\begin{gather}\label{ProblimJ}
\mathbb{P}(s,\alpha,\beta):=\lim_{n\rightarrow\infty}\mathbb{P}\left(\frac{s}{4n^2},\alpha,\beta,n\right),
\end{gather}
and
\begin{gather*}
\sigma(s):=s\frac{d}{ds}\log \mathbb{P}(s,\alpha,\beta).
\end{gather*}
Then $\sigma(s)$ satisfies the Jimbo-Miwa-Okamoto $\sigma$ form of $P_{III}$ (see (3.13) of \cite{Jimbo1982})
\begin{gather}\label{probeqJ}
\left(s \sigma''\right)^2 +\sigma'\left(4 \sigma'+1\right) \left(s \sigma'-\sigma\right)-\alpha^2  \left(\sigma'\right)^2=0,
\end{gather}
with the initial conditions
\begin{gather*}
\sigma(0)=\sigma'(0)=0.
\end{gather*}
\end{theorem}

\begin{proof}
From \eqref{PnDnJ} and \eqref{sigmandefJ}, we find
\begin{align}
\sigma_n(t)-d_2=&t(t-1)\frac{d}{dt}\log \mathbb{P}(t,\alpha,\beta,n)+d_1t\non\\
=&s\left(\frac{s}{4n^2}-1\right)\frac{d}{ds}\log \mathbb{P}\left(\frac{s}{4n^2},\alpha,\beta,n\right)-\frac{(2n+\alpha+\beta)^2}{4}\cdot\frac{s}{4n^2}\non\\
\rightarrow&-\sigma(s)-\frac{s}{4},\qquad n\rightarrow\infty.\label{sigmanlimJ}
\end{align}
Upon replacing $\sigma_n(t)$ by $-\sigma(s)-\frac{s}{4}+d_2$ and $t$ by $\frac{s}{4n^2}$ in \eqref{sigmanPviJ}, the coefficient of the highest order term in $n$ leads to the desired equation.

The initial condition $\sigma(0)=0$ follows from the fact $\sigma_n(0)=d_2$, and $\sigma'(0)=0$ is an immediate consequence of setting $s=0$.
\end{proof}
{\bf Remark} The non-appearance of the $\beta$ parameter in \eqref{probeqJ}, does not necessary imply the the non-appearance of $\beta$ in the asymptotic
expansion of $\mathbb{P}(s,\al,\beta)$

Before proceeding to the derivation of the expansion $\sigma(s)$, for large $s$,  we first point out the following important fact about $\sigma(s)$.
\begin{lemma}
$\sigma(s)$ has a lower bound
\begin{gather*}
\sigma(s)>-\frac{s}{4}.
\end{gather*}
\end{lemma}
\begin{proof}
From the initial conditions satisfies by $\sigma_n(t)$:
\begin{gather*}
\sigma_n(0)-d_2=0, \qquad\sigma_n'(0)=d_1<0.
\end{gather*}
Assuming the continuity of $\sigma_n(t)$ in $t$, for small $t$, we see that $\sigma_n(t)-d_2$ is negative when $t$ is sufficiently enough to $0$.
Hence, in view of \eqref{sigmanlimJ}, we get as $t\rightarrow0$,
\begin{gather*}
-\sigma(s)-\frac{s}{4}=\lim_{n\rightarrow\infty}(\sigma_n(t)-d_2)<0,
\end{gather*}
which completes the proof.
\end{proof}

{\bf Remark} We observe that \eqref{probeqJ} is the same as the differential equation \eqref{sigmaeq} of the Laguerre case.
For this reason, in view of the expansion formula for $\sigma(s)$ given by \eqref{HT0}, we assume in the Jacobi case
\begin{gather*}
\sigma(s)=\sum_{j=0}^{\infty}b_j s^{1-\frac{j}{2}},\qquad s\rightarrow\infty.
\end{gather*}
Inserting it into \eqref{probeqJ}, by comparing the corresponding coefficients on both sides, we find $b_0=-\frac{1}{4}$ and $b_1=\pm\frac{\alpha}{2}$.
Since $\sigma(s)>-\frac{s}{4}$, we choose $b_1=\frac{\alpha}{2}$.
Finally, we find \eqref{HT0} is also valid for $\sigma(s)=s\frac{d}{ds}\log \mathbb{P}(s,\alpha,\beta)$.
Therefore, the asymptotic expansion \eqref{probs} holds for $\mathbb{P}(s,\alpha,\beta)$ up to a constant.

\begin{theorem} As $s\rightarrow\infty$, the probability that all the eigenvalues in JUE are greater than $\frac{s}{4n^2}$ with $n$ large enough, has the following asymptotic expression:
\begin{equation}\label{probsJ}
\begin{aligned}
\log\mathbb{P}(s,\alpha,\beta)=&c_2(\alpha,\beta)-\frac{s}{4}+\alpha \sqrt{s}-\frac{\alpha^2}{4}\log s+\frac{\alpha }{8}s^{-\frac{1}{2}}+\frac{\alpha
   ^2}{16}s^{-1}
   +\left(\frac{\alpha^3}{24}+\frac{3\alpha}{128}\right)s^{-\frac{3}{2}}\\
&+\left(\frac{\alpha^4}{32}+\frac{9\alpha^2}{128}\right)s^{-2}
+\left(\frac{\alpha^5}{40}+\frac{9\alpha^3}{64}+\frac{45\alpha}{1024}\right)s^{-\frac{5}{2}}\\
&+\left(\frac{\alpha^6}{48}+\frac{15\alpha^4}{64}+\frac{9\alpha^2}{32}\right)s^{-3}+O\left(s^{-\frac{7}{2}}\right),
\end{aligned}
\end{equation}
where $c_2(\alpha,\beta)$ is shown later to be
\begin{gather*}
c_2(\alpha,\beta)=\log \left[\frac{G(\alpha+1)G^2(\beta+1)}{(2\pi)^{(\alpha+\beta)/2}}\right]+\frac{\beta(\beta-1)}{2}-(\beta+\tfrac{1}{2})\log\Gamma(\beta).
\end{gather*}
\end{theorem}

In view of \eqref{Pn0DnJ} and \eqref{Pn1DnJ}, it follows from \eqref{ProblimJ} and \eqref{probsJ} that
\begin{equation}\label{Pn0probJ}
\begin{aligned}
\lim_{n\rightarrow\infty}&\frac{P_{n}(0;\frac{s}{4n^2},\alpha,\beta)}{P_n(0;0,\alpha,\beta)}=\frac{\mathbb{P}(s,\alpha+1,\beta)}{\mathbb{P}(s,\alpha,\beta)}\\
&\sim \exp\left[c_2(\alpha+1,\beta)-c_2(\alpha,\beta)+\sqrt{s}-\left(\frac{\alpha}{2}+\frac{1}{4}\right)\log s+O\left(s^{-\frac{1}{2}}\right)\right],
\qquad s\rightarrow\infty,
\end{aligned}
\end{equation}
and
\begin{gather}\label{Pn1probJ}
\lim_{n\rightarrow\infty}\frac{P_{n}(1;\frac{s}{4n^2},\alpha,\beta)}{P_n(1;0,\alpha,\beta)}=\frac{\mathbb{P}(s,\alpha,\beta+1)}{\mathbb{P}(s,\alpha,\beta)}\sim {\rm e}^{c_2(\alpha,\beta+1)-c_2(\alpha,\beta)},\qquad s\rightarrow\infty.
\end{gather}

We now apply Dyson's Coulomb fluid method to derive the asymptotic formulas for $P_n(0;t,\alpha,\beta)$ and $P_n(1;t,\alpha,\beta)$.
The results combined with \eqref{Pn0probJ} and \eqref{Pn1probJ} will help us find $c_2(\alpha,\beta)$.

\subsection{The evaluation of $P_n(0;t,\alpha,\beta)$ and $P_n(1;t,\alpha,\beta)$ via Dyson's Coulomb fluid}
Putting
\begin{gather*}
v(x)=-\log w(x,\alpha,\beta)=-\alpha\log x-\beta\log(1-x),\quad\alpha>0,\quad\beta>0,
\end{gather*}
in \eqref{densitycon.} with $a=t$, i.e.
\begin{gather*}
\rho(x)=\frac{1}{2\pi^2}\sqrt{\frac{b-x}{x-t}}\,P\int_t^b \frac{ v'(y)}{y-x}\sqrt{\frac{y-t}{b-y}}dy,\qquad x\in(t,b),\quad b<1,
\end{gather*}
gives us the equilibrium density of the fluid,
\begin{gather*}
\rho(x)=\frac{1}{2\pi}\sqrt{\frac{b-x}{x-t}}\left[-\frac{\alpha}{x}\sqrt{\frac{t}{b}}+\frac{\beta}{1-x}\sqrt{\frac{1-t}{1-b}}\right],\;\;0<t<x<b<1.
\end{gather*}
An easy computation shows that
$$
\frac{d}{dx}\rho(x)\sqrt{\frac{x-t}{b-x}}>0,\;\; t<x<b,
$$
and $\beta\sqrt{tb}>\alpha\sqrt{(1-t)(1-b)}$ which implies $\rho(x)>0$ for $t<x<b$. By means of the integral identities listed in the Appendix, the normalization condition $\int_t^b\rho(x)dx=n$ reads
\begin{gather*}
n=-\frac{\alpha+\beta}{2}+\frac{\alpha}{2}\sqrt{\frac{t}{b}}+\frac{\beta}{2}\sqrt{\frac{1-t}{1-b}}.
\end{gather*}
This gives rise to an algebraic equation of degree four in $b$,
\begin{gather*}
\left[(2n+\alpha+\beta)^2b(1-b)-\alpha^2t(1-b)-\beta^2(1-t)b\right]^2=4\alpha^2\beta^2t(1-t)b(1-b),
\end{gather*}
from which it follows that
\begin{gather*}
b=1-\frac{\beta^2(1-t)}{4n^2}+O(n^{-3}),\qquad n\rightarrow\infty.
\end{gather*}
We do not display the $O(n^{-3})$ term as it would not affect the outcome.

We follow the methodology in the previous section for LUE. According to \eqref{PnS1S2general}, the monic orthogonal polynomials $P_n(z;t,\alpha,\beta)$ associated with $x^{\alpha}(1-x)^{\beta},x\in[t,1]$, evaluated at $z=0$, is given by
\begin{align*}
P_n(0;t,\alpha,\beta)&\sim {\rm e}^{-S_1(0;t,\alpha,\beta)-S_2(0;t,\alpha,\beta)},\qquad n\rightarrow\infty,
\end{align*}
where
\begin{align*}
{\rm e}^{-S_1(0;t,\alpha,\beta)}=&\frac{1}{2}\left[\left(\frac{b}{t}\right)^{\frac{1}{4}}
+\left(\frac{t}{b}\right)^{\frac{1}{4}}\right].
\end{align*}
With the aid of the integral identities in the Appendix, and by taking the branch $-t=t{\rm e}^{\pi i}$ and $-b=b {\rm e}^{\pi i}$, we obtain
\begin{equation*}
\begin{aligned}
{\rm e}^{-S_2(0;t,\alpha,\beta)}=&(-1)^n 2^{-2n-2\alpha-\beta}\left(\sqrt{t}+\sqrt{b}\right)^{2n+2\alpha+\beta}(tb)^{-\frac{\alpha}{2}}\\
&\cdot\left(1-\left(\sqrt{tb}-\sqrt{(1-t)(1-b)}\right)^2\right)^{-\frac{\beta}{2}}\left(\sqrt{1-t}+\sqrt{1-b}\right)^{\beta}.
\end{aligned}
\end{equation*}

Now we discuss the approximation of $P_n(z;t,\alpha,\beta)$ at $z=1$. Once again from
\eqref{PnS1S2general}, we obtain
\begin{align*}
P_n(1;t,\alpha,\beta)&\sim {\rm e}^{-S_1(1;t,\alpha,\beta)-S_2(1;t,\alpha,\beta)}, \qquad n\rightarrow\infty,
\end{align*}
where
\begin{align*}
{\rm e}^{-S_1(1;t,\alpha,\beta)}=&\frac{1}{2}\left[\left(\frac{1-b}{1-t}\right)^{\frac{1}{4}}
+\left(\frac{1-t}{1-b}\right)^{\frac{1}{4}}\right],
\end{align*}
and, by applying the integral identities in the Appendix,
\begin{equation*}
\begin{aligned}
{\rm e}^{-S_2(1;t,\alpha,\beta)}=&2^{-2n-\alpha-2\beta}\left(\sqrt{1-t}+\sqrt{1-b}\right)^{2n+\alpha+2\beta}\left(\sqrt{t}+\sqrt{b}\right)^{\alpha}\\
&\cdot\left(1-\left(\sqrt{tb}-\sqrt{(1-t)(1-b)}\right)^2\right)^{-\frac{\alpha}{2}}\left((1-t)(1-b)\right)^{-\frac{\beta}{2}}.
\end{aligned}
\end{equation*}

Finally, pooling together the above results, we give an evaluation of $P_n(0;t,\alpha,\beta)$ and $P_n(1;t,\alpha,\beta)$ as $n\rightarrow\infty$.
\begin{theorem} The monic polynomials $P_n(x;t,\alpha,\beta)$ orthogonal with respect to $x^{\alpha}(1-x)^{\beta},~\alpha>-1,\beta>0$, over $[t,1]$ are approximated at $x=0$ and $x=1$ as $n\rightarrow\infty$ ,by
\begin{align}
(-1)^nP_n(0;t,\alpha,\beta)\sim&2^{-2n-2\alpha-\beta-1}\left[t^{-\frac{1}{4}}
+t^{\frac{1}{4}}\right]\left[1+\sqrt{t}\right]^{2n+2\alpha+\beta}t^{-\frac{\alpha}{2}}\nonumber\\
\sim&2^{-2n-\alpha-\beta-\frac{1}{2}}{\rm e}^{\sqrt{s}}s^{-\frac{\alpha}{2}-\frac{1}{4}}n^{\alpha+\frac{1}{2}},\label{Pn0sJ}
\end{align}
and
\begin{align}
P_n(1;t,\alpha,\beta)\sim&2^{-2n-\alpha-\beta-\frac{1}{2}}\,n^{\beta+\frac{1}{2}}\,\beta^{-\beta-\frac{1}{2}}\,{\rm e}^{\beta}\,(1-t)^{n}(1+\sqrt{t})^{\alpha}\nonumber\\
\sim&2^{-2n-\alpha-\beta-\frac{1}{2}}\,n^{\beta+\frac{1}{2}}\,\beta^{-\beta-\frac{1}{2}}\,{\rm e}^{\beta},\label{Pn1sJ}
\end{align}
where $s=4n^2\,t$.
\end{theorem}

\begin{corollary} The constant $c_2(\alpha,\beta)$, appearing in \eqref{probsJ}, turns out to be
\begin{gather*}
c_2(\alpha,\beta)=\log \left[\frac{G(\alpha+1)G^2(\beta+1)}{(2\pi)^{(\alpha+\beta)/2}}\right]+\frac{\beta(\beta-1)}{2}-(\beta+\tfrac{1}{2})\log\Gamma(\beta),
\end{gather*}
where $G(\cdot)$ is the Barnes-G function.
\end{corollary}
\begin{proof} Again, in this proof, the symbol $\sim$ refers to 'asymptotic to' for large $n$.

From asymptotic approximation for Gamma function
\begin{gather*}
\Gamma(n+\alpha)\sim\Gamma(n)n^{\alpha}, \qquad\alpha\in\mathbb{\mathbb{C}},
\end{gather*}
and the Legendre duplication formula (formula (1.2.3), \cite{Lebedev})
\begin{gather*}
\sqrt{\pi}\Gamma(2n)=2^{2n-1}\Gamma(n)\Gamma(n+\tfrac{1}{2})\sim 2^{2n-1}\Gamma^2(n)\sqrt{n},
\end{gather*}
we obtain, by recalling \eqref{Pn00J},
\begin{align}
(-1)^nP_{n}(0;0,\alpha,\beta)=&\frac{\Gamma(n+\alpha+\beta+1)\Gamma(n+\alpha+1)}{\Gamma(2n+\alpha+\beta+1)\Gamma(\alpha+1)}\nonumber\\
=&\frac{\Gamma(n)\,n^{\alpha+\beta+1}\,\Gamma(n)\,n^{\alpha+1}}{\Gamma(2n)\,(2n)^{\alpha+\beta+1}\Gamma(\alpha+1)}\nonumber\\
\sim&\sqrt{\pi}2^{-2n-\alpha-\beta}\frac{n^{\alpha+\frac{1}{2}}}{\Gamma(\alpha+1)},\label{Pn00asy.J}
\end{align}
so that, from \eqref{Pn10J},
\begin{align}
P_n(1;0,\alpha,\beta)=&(-1)^nP_{n}(0;0,\beta,\alpha)\nonumber\\
\sim&\sqrt{\pi}2^{-2n-\alpha-\beta}\frac{n^{\beta+\frac{1}{2}}}{\Gamma(\beta+1)}.\label{Pn10asy.J}
\end{align}
Combining \eqref{Pn00asy.J} with \eqref{Pn0sJ}, and \eqref{Pn10asy.J} with \eqref{Pn1sJ}, we find
\begin{align}
\frac{P_n(0;\frac{s}{4n^2},\alpha,\beta)}{P_{n}(0;0,\alpha,\beta)}
\sim&\frac{\Gamma(\alpha+1)}{\sqrt{2\pi}}{\rm e}^{\sqrt{s}}s^{-\frac{\alpha}{2}-\frac{1}{4}}\nonumber\\
\sim&\exp\left[\log\left(\frac{\Gamma(\alpha+1)}{\sqrt{2\pi}}\right)+\sqrt{s}-\left(\frac{\alpha}{2}+\frac{1}{4}\right)\log s\right],\label{Pn0asy.J}
\end{align}
and, assuming $\beta>0$,
\begin{align}
\frac{P_n(1;\frac{s}{4n^2},\alpha,\beta)}{P_n(1;0,\alpha,\beta)}
\sim&\frac{\Gamma(\beta+1)}{\sqrt{2\pi}}{\rm e}^{\beta}\beta^{-\beta-\frac{1}{2}}\nonumber\\
=&\exp\left[\log\left(\frac{\Gamma(\beta+1)}{\sqrt{2\pi}}\right)+\beta-(\beta+\tfrac{1}{2})\log\beta\right].\label{Pn1asy.J}
\end{align}
Comparing \eqref{Pn0asy.J} with \eqref{Pn0probJ}, and \eqref{Pn1asy.J} with \eqref{Pn1probJ}, we arrive at
\begin{gather}\label{c2alphadiffJ}
c_2(\alpha+1,\beta)-c_2(\alpha,\beta)=\log\left(\tfrac{\Gamma(\alpha+1)}{\sqrt{2\pi}}\right),
\end{gather}
and
\begin{gather}\label{c2betadiffJ}
c_2(\alpha,\beta+1)-c_2(\alpha,\beta)=\log\left(\frac{\Gamma(\beta+1)}{\sqrt{2\pi}}\right)+\beta-(\beta+\tfrac{1}{2})\log\beta,
\end{gather}
respectively. Hence, it follows from \eqref{c2alphadiffJ}
\begin{gather*}
c_2(\alpha,\beta)=\log \frac{G(\alpha+1)}{(2\pi)^{\alpha/2}}+f(\beta),
\end{gather*}
and from \eqref{c2betadiffJ} we find a difference equation satisfied by $f(\beta)$:
\begin{gather*}
f(\beta+1)-f(\beta)=\log\left(\frac{\Gamma(\beta+1)}{\sqrt{2\pi}}\right)+\beta-(\beta+\tfrac{1}{2})\log\beta,
\end{gather*}
and consequently,
\begin{gather*}
f(\beta)=\log \frac{G^2(\beta+1)}{(2\pi)^{\beta/2}}+\frac{\beta(\beta-1)}{2}-(\beta+\tfrac{1}{2})\log\Gamma(\beta).
\end{gather*}
Therefore, we finally obtain
\begin{gather*}
c_2(\alpha,\beta)=\log \left[\frac{G(\alpha+1)G^2(\beta+1)}{(2\pi)^{(\alpha+\beta)/2}}\right]+\frac{\beta(\beta-1)}{2}-(\beta+\tfrac{1}{2})\log\Gamma(\beta).
\end{gather*}
\end{proof}
\begin{remark}
We note an integral representation for $f(\beta+1)-f(\beta)$:
\begin{align*}
f(\beta+1)-f(\beta)=&\log\Gamma(\beta)+\beta-(\beta-\tfrac{1}{2})\log\beta-\tfrac{1}{2}\log(2\pi)\\
=&\int_0^{\infty}\left(\frac{1}{2}-\frac{1}{t}+\frac{1}{e^t-1}\right)\frac{e^{-\beta t}}{t}dt,
\end{align*}
which results from Binet's formula (\cite{WhittakerWatson}, p.249). Hence, the constant $c_2(\gamma)$ that appears in \cite{ChenChenFanJMP2016} has an explicit evaluation:
\begin{gather*}
c_2(\gamma)=\log\frac{G(1+\gamma-\lambda-\beta)\,G(1+\gamma)}{(2\pi)^{\gamma/2}}+\frac{\gamma(\gamma-1)}{2}-(\gamma+\tfrac{1}{2})\log\Gamma(\gamma).
\end{gather*}
\end{remark}

\section{The Asymptotics of the Gap Probability of the Jacobi Unitary Ensembles}
The probability that the interval $(-a,a)$ has no eigenvalues in the (symmetric) Jacobi unitary ensembles with the weight
\begin{gather*}
w_0(x,\beta)=(1-x^2)^{\beta},\qquad x\in[-1,1],\quad\beta>0,
\end{gather*}
is given by
\begin{gather}\label{probnDna0J}
\mathbb{P}(a,\beta,n)=\frac{\det\left(\int_{-1}^1x^{i+j}w(x,a,\beta)dx\right)_{i,j=0}^{n-1}}{\det\left(\int_{-1}^1x^{i+j}w_0(x,\beta)dx\right)_{i,j=0}^{n-1}}
=:\frac{D_n(a,\beta)}{D_n(0,\beta)}.
\end{gather}
Here $w(x,a,\beta)$ is the discontinuous Jacobi weight with two jumps
\begin{align*}
w(x,a,\beta)=&(1-x^2)^{\beta}\theta(x^2-a^2),\qquad x\in[-1,1],\quad\beta>0,
\end{align*}
where $\theta(x)$ is 1 if $x\geq0$ and 0 otherwise.
We shall be concerned with the behavior of $\mathbb{P}(a,\beta,n)$ under double scaling.

We begin with the normalization relation of the monic polynomials orthogonal with respect to $w(x,a,\beta)$ over $[-1,1]$:
\begin{align}
h_j(a,\beta):=&\int_{-1}^{1}P_j^2(x;a,\beta)(1-x^2)^{\beta}\theta(x^2-a^2)dx\non\\
=&2\int_{0}^{1}P_j^2(x;a,\beta)(1-x^2)^{\beta}\theta(x^2-a^2)dx\label{hjJacobi}.
\end{align}
Note that $\{P_j\}$ can be normalized as \cite{Chihara}
\begin{gather*}
P_{2n}(x;a,\beta)=x^{2n}+p(2n,a,\beta)x^{2n-2}+\cdots+P_{2n}(0;a,\beta),
\end{gather*}
and
\begin{align*}
P_{2n+1}(x;a,\beta)=&x^{2n+1}+p(2n+1,a,\beta)x^{2n-1}+\cdots+{\rm const.}\,x\\
=&x\left(x^{2n}+p(2n+1,a,\beta)x^{2n-2}+\cdots+{\rm const.}\right).
\end{align*}

By introducing into \eqref{hjJacobi} the change of variable $x^2=t$, we find
\begin{align*}
h_{2n}(a,\beta)=&2\int_{0}^{1}P_{2n}^2(x;a,\beta)(1-x^2)^{\beta}\theta (x^2-a^2)dx\\
=&2\int_{0}^{1}P_{2n}^2(\sqrt{t};a,\beta)(1-t)^{\beta}\theta (t-a^2)\frac{dt}{2\sqrt{t}}\\
=&\int_{a^2}^{1}\widetilde{P}_n^2(t;a,\beta)t^{-\frac{1}{2}}(1-t)^{\beta}dt=:\widetilde{h}_n(a,\beta),\\
\intertext{and}
h_{2n+1}(a,\beta)=&2\int_{0}^{1}P_{2n+1}^2(x;a,\beta)(1-x^2)^{\beta}\theta (x^2-a^2)dx\\
=&2\int_{0}^{1}\left\{\sqrt{t}\,\widehat{P}_n(t;a,\beta)\right\}^2(1-t)^{\beta}\theta (t-a^2)\frac{dt}{2\sqrt{t}}\\
=&\int_{a^2}^{1}\widehat{P}_n^2(t;a,\beta)t^{\frac{1}{2}}(1-t)^{\beta}dt=:\widehat{h}_n(a,\beta).
\end{align*}
Here $\widetilde{P}_n(t;a,\beta)$ and $\widehat{P}_n(t;a,\beta)$ are monic polynomials of degree $n$ in the variable $t$, orthogonal with respect to $t^{-\frac{1}{2}}(1-t)^{\beta}$ and $t^{\frac{1}{2}}(1-t)^{\beta}$ over $[a^2,1]$ respectively.

If we define the Hankel determinants generated by $t^{-\frac{1}{2}}(1-t)^{\beta}$ and $t^{\frac{1}{2}}(1-t)^{\beta},t\in[a^2,1]$, by
\begin{align*}
\widetilde{D}_m(a,\beta):=&\det \left(\int_{a^2}^{1}t^{i+j}t^{-\frac{1}{2}}(1-t)^{\beta}dt\right)_{i,j=0}^{m-1}=\prod_{l=0}^{m-1}\widetilde{h}_l(a,\beta),\\
\widehat{D}_m(a,\beta):=&\det \left(\int_{a^2}^{1}t^{i+j}t^{\frac{1}{2}}(1-t)^{\beta}\right)_{i,j=0}^{m-1}=\prod_{l=0}^{m-1}\widehat{h}_l(a,\beta),
\end{align*}
then it follows that
\begin{align}
D_n(a,\beta):=&\det\left(\int_{-1}^1x^{i+j}w(x,a,\beta)dx\right)_{i,j=0}^{n-1}=\prod_{j=0}^{n-1}h_j(a,\beta)\nonumber\\
=&\begin{cases}
\widetilde{D}_{k+1}\,\widehat{D}_{k}, & n=2k+1,\\
\widetilde{D}_{k}\,\widehat{D}_{k}, & n=2k.\label{HankelapartJ}
\end{cases}
\end{align}

Based on \eqref{probsJ}, we establish in the next theorem the asymptotic expression for
\begin{gather*}
\mathbb{P}(b,\beta):=\lim_{n\rightarrow\infty}\mathbb{P}\left(a,\beta,n\right),\qquad b:=na,
\end{gather*}
under the assumption that $a\rightarrow0$ and $n\rightarrow\infty$ such that $b$ is fixed.
Note the double scaling is unlike the GUE case, where $\tau=2{\sqrt{2n}}\:a$
\begin{theorem} The probability that the interval $(-\tfrac{b}{n},\tfrac{b}{n}), n\rightarrow\infty$, is free of eigenvalues in the Jacobi unitary ensembles with weight $(1-x^2)^{\beta}, x\in[-1,1],\;\beta>0$ is approximated by
\begin{equation*}
\begin{aligned}
\log\mathbb{P}(b,\beta)=&-\frac{b^2}{2}-\frac{\log b}{4}+\log\left[G(\tfrac{1}{2})^2\sqrt{\pi}\right]+\log \left[\frac{G^4(\beta+1)}{(2\pi)^{\beta}}\right]\\
&+\beta(\beta-1)-(2\beta+1)\log\Gamma(\beta)\\
&+\frac{1}{32\,b^{2}}+\frac{5}{128\,b^{4}}+\frac{131}{768\,b^{6}}+O\left(b^{-8}\right), \qquad b\rightarrow\infty.
\end{aligned}
\end{equation*}
\end{theorem}
\begin{proof}
From \eqref{probnDna0J} and \eqref{HankelapartJ}, we find
\begin{align*}
\mathbb{P}(b,\beta)=&\lim_{n\rightarrow\infty}\mathbb{P}\left(a,\beta,n\right)=\lim_{n\rightarrow\infty}\frac{D_n\left(a,\beta\right)}{D_n(0,\beta)}\\
=&\lim_{k\rightarrow\infty}\frac{\widetilde{D}_{k}\left(a,\beta\right)}{\widetilde{D}_{k}(0,\beta)}
\cdot\lim_{k\rightarrow\infty}\frac{\widehat{D}_{k}\left(a,\beta\right)}{\widehat{D}_{k}(0,\beta)},
\end{align*}
so that, in view of \eqref{PnDnJ} and noting that $b^2=n^2a^2\sim4k^2a^2$ as $k\rightarrow\infty$,
\begin{align*}
\mathbb{P}(b,\beta)=&\lim_{k\rightarrow\infty}\mathbb{P}(a^2,-\tfrac{1}{2},\beta,k)\cdot\lim_{k\rightarrow\infty}\mathbb{P}(a^2,\tfrac{1}{2},\beta,k)\\
=&\mathbb{P}(b^2,-\tfrac{1}{2},\beta)\cdot\mathbb{P}(b^2,\tfrac{1}{2},\beta),
\end{align*}
Here $\mathbb{P}(a^2,\alpha,\beta, k)$ is the probability that all the eigenvalues of $k\times k$ Hermitian matrices with weight $x^{\alpha}(1-x)^{\beta}$ are greater than $a^2$ and $\mathbb{P}\left(b^2,\alpha,\beta\right)$ defined by \eqref{ProblimJ} is the scaled limiting probability of $\mathbb{P}(a^2,\alpha,\beta, k)$, i.e. $\mathbb{P}\left(b^2,\alpha,\beta\right)=\lim_{k\rightarrow\infty}\mathbb{P}\left(\frac{b^2}{4k^2},\alpha,\beta,k\right)$.

Therefore, according to \eqref{probsJ}, we establish
\begin{align*}
\log\mathbb{P}(b,\beta)=&\log\mathbb{P}(b^2,-\tfrac{1}{2},\beta)+\log\mathbb{P}(b^2,\tfrac{1}{2},\beta)\\
=&c_2\left(-\tfrac{1}{2},\beta\right)+c_2(\tfrac{1}{2},\beta)-\frac{b^2}{2}-\alpha^2\log b+\frac{\alpha
   ^2}{8}b^{-2}\\
&+\left(\frac{\alpha^4}{16}+\frac{9\alpha^2}{64}\right)b^{-4}
+\left(\frac{\alpha^6}{24}+\frac{15\alpha^4}{32}+\frac{9\alpha^2}{16}\right)b^{-6}\\
&+O\left(b^{-8}\right),\qquad b\rightarrow\infty.
\end{align*}
Here the constant term reads
\begin{align*}
c_2\left(-\tfrac{1}{2},\beta\right)&+c_2(\tfrac{1}{2},\beta)\\
=&\log\left[G(\tfrac{1}{2})^2\sqrt{\pi}\right]+\log \left[\frac{G^4(\beta+1)}{(2\pi)^{\beta}}\right]+\beta(\beta-1)-(2\beta+1)\log\Gamma(\beta),
\end{align*}
where the properties $G(\tfrac{3}{2})=G(\tfrac{1}{2})\Gamma(\tfrac{1}{2})$ and $\Gamma(\frac{1}{2})=\sqrt{\pi}$ are used.
The proof is completed.
\end{proof}
{\section {Conclusion}}
We obtained in this paper the constant term in the asymptotic expansion of the large gap probability of the Gaussian unitary ensembles;
reproducing the Widom-Dyson constant. This is done through the study of the smallest eigenvalue distribution of the Laguerre
unitary ensembles, and specializing $\alpha=\pm 1/2, $ in the constant obtained. Finally, we derive the asymptotic expansion of the smallest eigenvalue
distribution of the Jacobi unitary ensembles $(x^{\alpha}(1-x)^{\beta}$,\;\;$x\in(0,1),\;\;\alpha>-1,\;\beta>0)$. In this situation, although the double
scaled $\sigma$ equation is identical with the LUE $\sigma-$ equation, however the constant in the $JUB$ problem depends on $\beta$.

\section*{Appendix: Some Relevant Integral Identities}
We list here some integrals, which are relevant to our derivation and can be found in \cite{ChenHaqMckay2013}, \cite{ChenMcKay2012} and \cite{GradshteynRyzhik2007}. The basic assumption is that $0<a<b$.
\begin{align*}
\int_{a}^b\frac{dx}{\sqrt{(b-x)(x-a)}}=&\pi,\\
\int_{a}^b\frac{x\; dx}{\sqrt{(b-x)(x-a)}}=&\frac{(a+b)\pi}{2},\\
\int_{a}^b\frac{dx}{x\sqrt{(b-x)(x-a)}}=&\frac{\pi}{\sqrt{ab}},\\
\int_{a}^b\frac{dx}{x^2\sqrt{(b-x)(x-a)}}=&\frac{(a+b)\pi}{2(ab)^{\frac{3}{2}}},\\
\int_{a}^b\frac{dx}{(1-x)\sqrt{(b-x)(x-a)}}=&\frac{\pi}{\sqrt{(1-a)(1-b)}},\;\;(b<1).
\end{align*}
\begin{align*}
\int_{a}^b\frac{\log (1-x)\;dx}{\sqrt{(b-x)(x-a)}}=&2\pi\log\left[\frac{\sqrt{1-a}+\sqrt{1-b}}{2}\right],\;\;(b<1)\\
\int_{a}^b\frac{\log (1-x)\;dx}{x\sqrt{(b-x)(x-a)}}=&\frac{\pi}{\sqrt{ab}}\log\left[\tfrac{1-\left(\sqrt{ab}-\sqrt{(1-a)(1-b)}\right)^2}{(\sqrt{a}+\sqrt{b})^2}\right],
\;\;b<1\\
\int_{a}^b\frac{\log x\;dx}{\sqrt{(b-x)(x-a)}}=&2\pi\log\left[\frac{\sqrt{a}+\sqrt{b}}{2}\right],\\
\int_{a}^b\frac{\log x\;dx}{x\sqrt{(b-x)(x-a)}}=&\frac{2\pi}{\sqrt{ab}}
\log\left[\frac{2\sqrt{ab}}{\sqrt{a}+\sqrt{b}}\right],
\end{align*}
\begin{align*}
\int_{a}^b\frac{\log x\;dx}{(x-1)\sqrt{(b-x)(x-a)}}=&\pi\frac{\log\left[\tfrac{(\sqrt{1-a}+\sqrt{1-b})^2}{1-\left(\sqrt{ab}-\sqrt{(1-a)(1-b)}\right)^2}\right]}{\sqrt{(1-a)(1-b)}},\qquad (b<1),\\
\int_{a}^b\frac{\log (1-x)\;dx}{(x-1)\sqrt{(b-x)(x-a)}}=&2\pi\frac{\log \left[\frac{1}{2\sqrt{1-a}}+\frac{1}{2\sqrt{1-b}}\right]}{\sqrt{(1-a)(1-b)}},\qquad (b<1).
\end{align*}

\section*{Acknowledgments}
The authors gratefully acknowledge the generous support of
Macau Science and Technology Development Fund under the grant numbers FDCT 130/2014/A3, and FDCT 023/2017/A1,
and the University of Macau through
MYRG2014-00011-FST, MYRG2014-00004-FST  and  the National Science Foundation of China under Project No.11671095.

\end{document}